\documentclass{article}

\usepackage[utf8]{inputenc}
\usepackage{amsmath}
\usepackage{amsthm}
\usepackage{amssymb}
\usepackage{amsfonts}
\usepackage{graphicx}
\usepackage{url}
\usepackage{color}
\usepackage{array}
\usepackage{cases}
\usepackage{epsfig}
\usepackage{times}
\usepackage{fullpage}
\usepackage{tikz}
\usepackage{natbib}
\usepackage{pgf} 
\usepackage{fp}
\usepackage{algorithm}
\usepackage{algorithmic}
\usepackage{comment}
\usepackage{enumitem}
\usepackage{pifont}
\usepackage[mathscr]{euscript}
\usepackage{textcomp,gensymb}

\begin{document}

\title{Satisfactory Budget Division\thanks{A preliminary version of this article appears in the proceedings of AAMAS
2025 as an extended abstract \cite{GourvesLMP25}.}}
\author{Laurent Gourv{\`e}s$^{1}$ 
\thanks{Corresponding author} 
\and Michael Lampis$^2$ \and  Nikolaos Melissinos$^3$\and Aris Pagourtzis$^4$}

\date{{\small 1. Universit\'e Paris-Dauphine, Universit\'e PSL, CNRS, LAMSADE, 75016, Paris, France\\
\texttt{laurent.gourves@lamsade.dauphine.fr}, \texttt{ORCID 0000-0002-5076-1583}\\
2. Universit\'e Paris-Dauphine, Universit\'e PSL, CNRS, LAMSADE, 75016, Paris, France\\
\texttt{michail.lampis@lamsade.dauphine.fr}, \texttt{ORCID 0000-0002-5791-0887}\\
3. Charles University, Computer Science Institute, Faculty of Mathematics and Physics, Prague, Czech Republic\\
\texttt{melissinos@iuuk.mff.cuni.cz}, \texttt{ORCID 0000-0002-0864-9803}\\
4. National Technical University of Athens \& Archimedes, Athena RC, Athens, Greece\\
\texttt{pagour@cs.ntua.gr}, \texttt{ORCID 0000-0002-6220-3722}
}
}

         
\newcommand{\BibTeX}{\rm B\kern-.05em{\sc i\kern-.025em b}\kern-.08em\TeX}

\def\dell{{\sf d}}
\def\D{{\sf D}}
\def\cost{\textbf{cost}}
\def\AAS{{\sc all-agents-sat}}

\def\mecha{coverage maximizer}
\def\mechalong{agent demand coverage maximizer}

\setlength\extrarowheight{5pt}

\newcommand{\xmark}{\ding{55}}%
\newcommand{\cmark}{\ding{51}}%

\newcommand{\commentLaurent}[1]{\textcolor{blue}{(L: #1)}}

\newcommand{\commentAris}[1]{\textcolor{magenta}{(AP: #1)}}

\newtheorem{instance}{Example}

\newenvironment{instancecont}[1]{%
  \par\medskip\noindent\textbf{Example~\ref{#1}. }\itshape
}{%
  \par\medskip
}

\newtheorem{definition}{Definition}
\newtheorem{observation}{Observation}
\newtheorem{theorem}{Theorem}
\newtheorem{lemma}{Lemma}
\newtheorem{proposition}{Proposition}


\maketitle 


\begin{abstract}
A divisible budget must be allocated to several projects, and agents are asked for their opinion on how much they would give to each project. 
We consider that an agent is satisfied by a  division of the budget if, for at least a certain predefined number $\tau$ of projects, the part of the budget actually allocated to each project is at 
least as large as the amount the agent requested. 
The objective is to find  
a budget division that ``best satisfies'' the agents.  
In this context, different problems can be stated and we  
address the following ones. We study  $(i)$ the largest proportion of agents that can be satisfied for any instance, $(ii)$  classes of instances admitting a budget division that satisfies all agents, $(iii)$ the complexity of deciding if, for a given instance, every agent can be satisfied, and finally $(iv)$ the question of finding, for a given instance, the smallest total budget 
to satisfy all agents.  
We provide answers to these complementary questions for several natural  
values of the parameter $\tau$, capturing scenarios where we seek to satisfy for each agent all; almost all; half; or at least one of her requests.

\end{abstract}

\begin{center}
\end{center}

\section{Introduction}

For many years, social choice has concerned the AI community, particularly for the computational questions that it generates \cite{BCELP2016,UE17,JR2024}. The central question in computational social choice is to take into account the individual preferences of several agents when developing a compromise solution. The main topics of this field are voting methods to elect representatives (e.g., committees), fair allocation of (possibly indivisible) goods and chores, or the partitioning of agents into stable subgroups (e.g., matchings, hedonic games). In addition to these fundamental topics which are now widely discussed in the literature, significant interest has recently arisen around collective budget issues \cite{BogomolnaiaMS05,ABM20,AS21,MPS20,AACKLP23,BGPSW24,FPPV21,CGL25,CGL25bis,WagnerM23,BRANDL2022102585,GoyalSSG23,CCP24,FS24,EGLST26}. Numerous models have been proposed and studied, including the now famous and well-studied \emph{participatory budgeting} \cite{Cabannes2004,de2022international,RSM25}. Budgeting problems 
address the recurring question of how to properly use a common budget for funding a given set of projects. This article is part of this vibrant trend. Its aim is to contribute, using a new concept of agent satisfaction, to knowledge on the existence and computation of an acceptable collective budget.

\subsection{The Model}

We are given a perfectly divisible 
budget of (normalized) value 1,  $m$ projects, 
and a set of agents $N=[n]$.\footnote{For every positive integer $k$, $[k]$  denotes the set $\{1, \ldots,k\}$.}  
Each agent $i \in N$ has reported a demand $\dell^i_j\in [0,1]$ for every project $j \in [m]$. These quantities are opinions on how to spend the budget, i.e., agent $i$  would devote $\dell^i_j$ to the project $j$ if she was the only decision maker. 

Let us clarify the semantics of the demands: a small (resp., large) demand does not mean that the agent considers the associated project to be of no (resp., great) interest. On the contrary, a demand of $\dell^i_j$ means that after an investment of (at least) $\dell^i_j$ on project $j$, agent $i$ judges the status of the project $j$ to be satisfactory. Therefore, a small (resp., large) demand from an agent means that she is quite satisfied (resp., not very satisfied) with the current state of the project and that a small (resp., large) part of the budget would make it acceptable.

A solution (a.k.a. budget division)  ${\bf x}$ is an element of $[0,1]^m$, and ${\bf x}$ is said to be  budget-feasible if $\sum_{j=1}^m x_j \le 1$. Here, $x_j$ is the $j$-th coordinate of ${\bf x}$ and it indicates how much of the common budget is actually spent on project $j$. We suppose that 
\begin{equation} \label{mass_l}
\sum_{j=1}^m \dell^i_j \le 1, \quad \forall i \in N
\end{equation}
holds in order to express that the demands of every fixed agent are compatible with a feasible division of the budget.

The model has many applications. It captures situations where an organization (e.g., a library, a city council, a company, a university, etc.) invests in different projects (or activities, topics, facilities, etc.) and the 
problem is to divide the budget in a way that satisfies as much as possible the members of the organization.

Our approach is to consider that an agent $i$ is {\em locally satisfied} by ${\bf x}$ for a given project $j$ if  $x_j \ge \dell^i_j$, i.e., enough resources are put on the project from the agent's perspective (and investing more than the agent's demand is not harmful). More globally, an agent is said to be {\em satisfied}  (or {\em covered}) by ${\bf x}$ if she is locally satisfied for at least $\tau$ projects, where the threshold $\tau \in [1,m]$ is a parameter belonging to the problem's input. Thus, a solution can satisfy several agents by satisfying them locally, possibly for different projects. 




Let us illustrate the setting with an example in which the demands are gathered in an $n \times m$ matrix $\D$ where the entry at line $i$ and column $j$ is equal to $\dell^i_j$.
\begin{instance} \label{instance_nocover} A multimedia library has 3 kinds of documents (book, DVD, and record), and   
4 employees (Alice, Bob, Carl, and Diana) who made the following demands concerning the purchase of new items. 
$$\left( \begin{array}{ccc}
0.5 & 0.5 & 0 \\
0 & 0.5 & 0.5 \\
0.6 & 0.1 & 0.3 \\
0.3 & 0.1 & 0.6 
\end{array} \right)$$
Alice and Bob agree that at least half of the budget should be devoted to new DVDs. Alice thinks that the other $50\%$ should be spent on new books, and nothing for records because there are enough albums on the shelves. However, Bob believes that the rest should be invested on new records, and nothing for books because there are quite enough books. Carl's opinion is to spend $60\%$ of the budget on new books, $10\%$ on new DVDs, and $30\%$ on new records. Finally, Diana prefers to devote $30\%$ of the budget on new books, $10\%$ on new DVDs, and $60\%$ on new records.   

Suppose $\tau=2$ and ${\bf x}=(0.3,0.6,0.1)$. Alice is satisfied by ${\bf x}$ because after the purchase of new items, the outcome meets her expectations concerning the DVD and record sections. Bob and Diana are also satisfied because enough money is invested on new books and DVDs. However Carl is not satisfied by ${\bf x}$ because even after the purchase of new items, the sections of books and records are below his expectations.
\end{instance}


A typical instance of the proposed budget division problem contains multiple agents who have heterogeneous demands for the projects. Then, what solution do we decide to implement, seeking to best satisfy agents? The aim of this article is to provide several complementary approaches to answer this question.

It is often impossible to satisfy all agents, as in Example \ref{instance_nocover} when  $\tau=2$ (see the justification in the appendix,  Section \ref{sec:sm:1}). Of course, the parameter $\tau$ plays a central role in this matter: the larger $\tau$ is, the more constraints we impose on ${\bf x}$ 
so that it satisfies an agent. 
 We mainly consider four  values of $\tau$ corresponding to four scenarios: $1$, $m/2$, $m-1$, and $m$. The cases  $\tau=m/2$ and $\tau=m-1$ are termed \emph{at-least-half scenario} 
and \emph{all-but-one scenario}, respectively. The other scenarios deal with extreme values of the threshold: $\tau \in \{1,m\}$. 

The values of $\tau$ that we consider 
range from 1, which is very undemanding, to $m$, which is very demanding. 
A positive result for the demanding case $\tau=m$ is valuable, but if it is out of reach, then maybe the problem is amenable to a small relaxation; that is why we also consider $\tau=m-c$ where $c$ is constant. Relaxations up to a constant such as EF1 are typical in fair division  \cite{AABFLMVW23}. A related concept dedicated to participatory budgeting called \emph{budget balanced up to one} has been recently introduced \cite{ALSVW24}. The choice of $\tau=m/2$ corresponds to an intermediate case between $1$ and $m$, where it is natural and reasonable to admit that an agent is globally satisfied if she is locally satisfied for a \emph{majority} of projects. The notion of majority appears to be central in many works on opinion formation (see for example \cite{AulettaFG18,Zehmakan21} and the references therein).

\subsection{Contributions and Organization}

We address the following complementary questions and provide the answers indicated. 

\begin{itemize}
\item

What proportion of agents can be satisfied for {\em any} instance? Our findings, given in Section \ref{sec:fraction} and summarized in Table \ref{tab_rho_HALF}, are bounds on the largest  fraction of agents that can be satisfied in any instance, for every scenario. It turns out that $\tau=1$ is the only case where all agents are satisfied. Interestingly, at least half of the population of agents can be covered when $\tau=m/2$, and such a result is obtained using a solution that can be constructed very easily.

\item

Which classes of instances admit a solution that satisfies all agents? For every scenario we characterize in Section \ref{sec:sat_all_agents} the values of $(n,m)$ for which every instance with $n$ agents and $m$ projects admits a budget-feasible solution satisfying all agents. The results, given in Section \ref{sec:sat_all_agents}, are summarized in Tables \ref{tab_HALF} and \ref{tab_ABO} for $\tau \in \{m/2,m-1\}$.

\item   
For a given instance, what is the difficulty of deciding whether all agents can be satisfied? 
We show in Section \ref{sec:social_welfare_max} that the problem's complexity has a somewhat counterintuitive behavior. 
If we allow $\sum_{j \in [m]} \dell^i_j <1$, then the problem is strongly NP-complete for $\tau=m-1$, but if we have $\sum_{j \in [m]} \dell^i_j =1$, then the problem can be resolved in pseudopolynomial time.  Surprisingly, the subtle distinction only plays a role when $\tau=m-1$ because the problem is shown strongly NP-complete for $\tau=m-c$ and all constant $c\ge 2$. Regarding the other values of $\tau$, the decision problem is polynomial when $\tau \in \{1,m\}$. The case $\tau=m/2$ is an intriguing question that we leave open.

\item What is the smallest total budget 
(possibly smaller or larger than $1$) 
to satisfy all agents? Our results, given in Section \ref{sec:min_tot_bud} and summarized in Table \ref{tab_total_budget}, follow two approaches: answering the question for a single instance or for an entire class of instances. In a nutshell, minimizing the budget for a given instance is polynomial time solvable if $\tau=m$, and NP-hard for $\tau \in \{1,m/2,m-1\}$. Concerning the question of the worst-case budget necessary for satisfying everybody, for a whole class of instances, we show that twice of the initial budget is enough when $\tau \le m/2$. However, one might need a budget of up to $m/2$ if $\tau \ge m-1$.       
\end{itemize}

Throughout the article, we will often use the fact that one can suppose 
w.l.o.g. that all the coordinates of a solution ${\bf x}$ appear in the agents'  demands (with the possible exception of zeros).

\begin{observation} \label{obs2} For all $\tau \in [1,m]$, from a solution ${\bf x}$ satisfying a set of agents $S \subseteq N$, 
one can build another solution ${\bf z}$ satisfying $S$ as well, such that $z_j \in \{0\} \cup \{\dell^i_j: i\in N\}$ for every $j \in [m]$, and $\sum_{j=1}^m z_j \leq  \sum_{j=1}^m x_j$.   
\end{observation}
\begin{proof} For each $j \in [m]$, proceed as follows. If $x_j \in \{0\} \cup \{\dell^i_j: i\in N\}$, then $z_j=x_j$. Otherwise, $z_j$ is given the largest value in $\{0\} \cup \{\dell^i_j: i\in N\}$ which does not exceed $x_j$. Thus, it cannot be $z_j < \dell_j^i \le x_j$ for any agent-project pair $(i,j) \in N \times [m]$, meaning that agents satisfied by ${\bf x}$ are also satisfied by ${\bf z}$.      
\end{proof}

\subsection{Related Work}

Since the early work of Bogomolnaia et al. \cite{BogomolnaiaMS05}, and later addressed in \cite{ABM20}, many articles have been written on allocating a divisible budget to a given set of projects. The problem has many applications: how much time (resp., space) 
is given to various topics in a given conference (resp., textbook),  which voting power is given to the parties composing a parliament, how much money is allocated to different charities by a company, etc.

Budget divisions are constructed on the basis of preferences expressed by some agents (a.k.a. voters). Depending on the context, the format used to express preferences and the way they are aggregated vary. Also, the decision about every given project can be binary (fund it entirely or not at all), or it can be a real 
within a given range (e.g., allocating a percentage of the total budget).     
Let us review some previous related works that fall into the collective budget framework.   


In the \emph{portioning} problem \cite{MPS20,FPPV21,AACKLP23,GoyalSSG23,CCP24,FS24,EGLST26}, we are given a perfectly divisible budget of one unit to be spent on a set of $m$ projects, and some information about the preferences of $n$ voters. A solution is a vector ${\bf x} \in [0,1]^m$ satisfying $\sum_{j=1}^m x_j=1$. An \emph{aggregation mechanism} (mechanism in short) is a function whose input and output are the preferences of the voters and a solution, respectively.

In the seminal article of Freeman et al. \cite{FPPV21} on portioning, every agent $i \in [n]$ declares a score vector $\textbf{s}^i=(s^i_1, \ldots,s^i_m)$ such that $\sum_{j=1}^m s^i_j=1$. The cost of agent $i$ for the solution ${\bf x}$  is the $L_1$-distance between $\textbf{x}$ and $\textbf{s}^i$, i.e.,  $dist_1(\textbf{x},\textbf{s}^i):=\sum_{j \in [m]} |s^i_j - x_j|$. The authors are interested in \emph{strategyproof} mechanisms for which no voter  should be able to change the output in her favor by misreporting her score vector. They give a class of strategyproof mechanisms named \emph{moving phantom} which leverages Moulin's method \cite{moulin1980strategy}. Within this class, they identify a specific mechanism, called \emph{independent markets}, which satisfies an extra fairness constraint named \emph{proportionality}.

In \cite{EGLST26}, Elkind et al. study the independent markets mechanism and two families of mechanisms from an axiomatic viewpoint. The first family is said to be coordinate-wise, in the sense that it produces a solution $\textbf{x}$ such that $x_j$ is a function of the multiset $\{s^i_j \mid i \in [n]\}$ for each $j \in [m]$. The second family defines $\textbf{x}$ as an optimum of a social welfare function. The desirable axioms under consideration include strategyproofness, proportionality, monotonicity and Pareto optimality, to name a few.  

Caragiannis et al. \cite{CCP24} also consider the portioning problem. The focus is on strategyproof mechanisms whose $L_1$-distance between their output and the mean vector $(\sum_{i\in [n]}s^i_j/n)_{j \in [m]}$ should be upper bounded by some quantity $\alpha$. They give strategyproof mechanisms such that $\alpha \leq 1/2$ when $m=2$, and $\alpha \leq 2/3+\epsilon$ when $m=3$. Going further, Freeman and Schmidt{-}Kraepelin \cite{FS24} proposed to also consider the $L_\infty$-distance, as a measure of fairness between the projects \cite{FS24}.

In the article by Michorzewski et al. \cite{MPS20} on portioning, 
voters have separately declared the subset of projects that they approve. The utility of some voter $i$ for the solution ${\bf x}$ is $u_i({\bf x})=\sum_{j\in A_i} x_j$ where $A_i \subseteq [m]$ is the set of projects approved by voter $i$. The authors consider the worst-case deterioration of the utilitarian social welfare of various (probabilistic) voting rules (e.g., Nash social welfare maximization) and the fairness axioms (e.g., individual fair share) that they guarantee.

In the work of Airiau et al. \cite{AACKLP23} on portioning, every agent  has ranked the projects by order of preference. Each ranking is converted into scores for the projects (using, for example, plurality, veto, or Borda), and the scores are used for evaluating budget divisions. Then, a solution maximizing the social welfare is sought, considering  various notions of social welfare: utilitarian, egalitarian, or Nash. This process gives rise to various {\em positional social decision schemes}. The authors have analyzed and evaluated the 
complexity of those schemes, together with the fairness axioms that they possibly satisfy. 

In \cite{GoyalSSG23}, Goyal et al. study the distortion of randomized mechanisms with low sample-complexity 
for the portioning problem where the social cost of an output $\textbf{x}$ is $\sum_{i=1}^n dist_1(\textbf{x},\textbf{s}^i)$. Their main result is a mechanism which uses 3 votes, the distortion of which is between $1.38$ and $1.66$.

Brill et al. \cite{BGPSW24} consider elections in which voters express approval votes over parties and a given number of seats must be distributed among the parties. This model is related to the previous budget division problems (the budget is the total number of seats, and the parties are the projects), with the specificity that parties get an integral number of seats.

In the model introduced and studied by Wagner and Meir \cite{WagnerM23}, every agent $i$ expresses an amount $t_i$ for modification of an initial budget $B_0$, namely agent $i$ prescribes a budget of $B_0+n t_i$ where $t_i >-B_0/n$. On top of that, every agent declares an allocation of the resulting budget on $m$ given projects. Wagner and Meir show that one can construct a VCG-like mechanism in this setting, preventing any agent's incentive to report false declarations. Another incentive-compatible mechanism has been proposed in a slightly different setting where there is no single exogenous budget, but each agent $i$ has her own budget $B_i$ \cite{BRANDL2022102585}. The problem is to determine the agents’ contributions, and 
distribute them on $m$ given projects while taking into account the agents utilities.

As mentioned previously, participatory budgeting is a  paradigm in which citizens can propose projects of public interest (e.g., refurbishing a school) and then vote to determine which of these projects are actually financed with public funds \cite{AS21,de2022international,RSM25}.  Projects are indivisible, meaning they are either fully funded or not funded at all.\footnote{In their survey, Rey et al. \cite{RSM25} discuss the divisible case of participatory budgeting, and their references mainly point to the portioning problem.} Similarly, in the model of Cardi et al., there is a common budget to be used for funding the ``private'' projects submitted by some agents  \cite{CGL25,CGL25bis}. A valid solution is a subset of projects whose total cost does not exceed the budget, and each agent's utility for a given solution is equal to the cost of the accepted projects submitted by her. In this setting, various notions of fairness and efficiency, together with their possible combination, have been addressed.

The previously mentioned works share important similarities with the model studied in this article. However, the specificity of our model is mainly based on the way of satisfying the agents, making our results not directly comparable with those of the literature. Indeed, previous works assume that every agent $i$ has a cardinal (dis)-utility with respect to a solution. In the present article, the agents have two possible statuses, depending on whether they are satisfied or not by the outcome. In addition, in contrast with previous works, we assume that each agent is equally interested in \emph{all} projects, and demands represent additional resources needed to make the status of existing projects acceptable (according to the agent's opinion). If an agent is satisfied with the current state of a project, then no need to devote extra resources from her point of view. On the contrary, if an agent is not satisfied by the state of a project, then some efforts are necessary and she would allocate new resources to it in order to make it acceptable. Thus, the objective of the present work is to reach global satisfaction of an agent after the extra resources from the budget are allocated.

\section{Guaranteed Fraction of Satisfied Agents} \label{sec:fraction}

We are interested in determining the largest proportion of the whole set of agents that can be satisfied with a single budget-feasible solution, in any case. The objective is to give lower and upper bounds on $\rho \in (0,1]$ such that \emph{every} instance with $n$ agents and $m$ projects admits a solution ${\bf x}$, possibly dependent on the instance, which satisfies at least $\rho n$ agents, and $\sum_{j=1}^m x_j \le 1$. Note that since in this section (and Section \ref{sec:sat_all_agents}) we are interested in proving bounds guaranteed to apply to all instances, we can without loss of generality focus on instances where the  
demands of each agent sum up to exactly $1$. Our results for each $\tau$ in $\{1,m/2,m-1,m\} $ are summarized in Table \ref{tab_rho_HALF}.

\begin{table}
\begin{center}
\caption{\label{tab_rho_HALF} Satisfying $\rho n$ agents}
\begin{tabular}{l|cccc}
$\tau$&$1$&$m/2$&$m-1$&$m$\\
\hline
$\rho$&1&$[\frac{1}{2}+\frac{1}{2n}, \frac{2}{3}+\frac{1}{n}]$&$1/n$& $1/n$\\
\end{tabular}
\end{center}
\end{table}


\subsection{At-Least-Half Scenario}

\label{sec:wcf:half}





In this section, we assume that $\tau$ is equal to $m/2$. We propose a rule which consists in 
following the recommendations of a specific agent.  Concretely, 
the {\sc \mechalong} rule ({\sc \mecha} for short) outputs the demand vector of an agent which satisfies a maximum number of agents (break ties arbitrarily). 

Our analysis relies on the following notion of coverage.

\begin{definition} \label{def_cov}
An $m$-vector ${\bf x}$ $\tau$-\emph{covers} another $m$-vector ${\bf y}$ when $x_j \ge y_j$ holds for at least $\tau$ 
distinct coordinates $j$.
\end{definition}

\begin{theorem} \label{thm_dict} When  $\tau=m/2$, every instance admits an agent whose demand vector satisfies at least $\frac{n+1}{2}$ agents, and this lower bound is tight.
\end{theorem}
\begin{proof} Let us first give an intermediate result.

\begin{lemma} \label{lem_pair} When $\tau=m/2$, for any positive number of projects, and any pair of agents, the demand vector of one agent of the pair $\tau$-covers the demand vector of the other agent of the pair.  
\end{lemma}

\begin{proof} Let $(\dell^1_1, \ldots, \dell^1_m)$ and $(\dell^2_1, \ldots, \dell^2_m)$ be the demand vectors of the agents. If $(\dell^1_1, \ldots, \dell^1_m)$ covers $(\dell^2_1, \ldots, \dell^2_m)$, then there exists an index set $J$ such that $|J|\ge m/2$ and $\dell^1_j \ge \dell^2_j$ holds for all $j \in J$. If $(\dell^1_1, \ldots, \dell^1_m)$ does not cover $(\dell^2_1, \ldots, \dell^2_m)$, then  
for at least one index set $J$ with  $|J|\ge m/2$, $\dell^1_j < \dell^2_j$ holds for all $j \in J$, which means that $(\dell^2_1, \ldots, \dell^2_m)$ covers $(\dell^1_1, \ldots, \dell^1_m)$. \end{proof}

Consider a complete undirected graph on $n$ vertices where every vertex is an agent.  For every edge $(i,j)$ in the graph, one point is given to $i$ when the demand vector of $i$ $\tau$-covers the demand vector of $j$. Using Lemma \ref{lem_pair}, at least one point is distributed for every edge. Thus, a total of at least $\frac{n(n-1)}{2}$ points is given, and there must be one agent in the graph, say $d$, whose number of points is at least $\frac{n-1}{2}$. Since an agent always $\tau$-covers herself, the demand vector of $d$ $\tau$-covers at least $\frac{n-1}{2}+1=\frac{n+1}{2}$ agents in total.      


To conclude, let us show that the lower bound on $\rho$ of $\frac{1}{2}+\frac{1}{2n}$ for the {\sc \mecha} rule is tight.
We start by selecting, arbitrarily, an odd number $m = n \ge 3$. Then, we consider the following instance.
\begin{equation} \label{example_tight_dictator} \left( \begin{array}{ccccc}
\frac{1}{2}+ \frac{1}{2^m} & \frac{1}{2^2} & \ldots & \frac{1}{2^{m-1}} & \frac{1}{2^m} \\
\frac{1}{2^m} & \frac{1}{2}+ \frac{1}{2^m} & \ldots & \frac{1}{2^{m-2}} & \frac{1}{2^{m-1}} \\
\cdots & \cdots & \ldots & \cdots & \cdots \\
\frac{1}{2^2} & \frac{1}{2^3} & \ldots & \frac{1}{2^m} & \frac{1}{2} + \frac{1}{2^m} 
\end{array} \right)
\end{equation}
We will show that there is no agent whose demand vector $\tau$-covers more than $1/2 + 1/2n$ demand vectors.
Due to the symmetry of the instance, we can assume that the {\mecha} is the first agent. Notice that the demand vector of the $i^{\text{th}}$ agent, where $i \neq 1$, is $\tau$-covered by the demand vector of the $1^{\text{st}}$ agent if and only if $i \ge (m+3)/2$. Indeed, for any $i \in {2, \ldots, m}$, we have that:
\begin{align*}
\text{for any project } j\ge i: &\ \dell_j^i \ge (1/2)^{j-i+1} > (1/2)^{j}= \dell_j^1, \text{ and} \\
\text{for any project } j < i: &\
\dell_j^i = (1/2)^{m-(i-j)+1} < (1/2)^{j}= \dell_j^1
\end{align*}
Therefore, the $i^{\text{th}}$ agent, where $i > 1$, is satisfied if and only if $i - 1 \ge m/2$. Thus, the demand vector of the first agent satisfies herself and every agent $i$, where $i \ge (m+3)/2$ (since $m$ is odd).
Therefore, $\rho =\big( m-(\frac{m+1}{2})+1\big)/m = \frac{1}{2} + \frac{1}{2m} = \frac{1}{2} + \frac{1}{2n}$.

The previous instance can be extended so that 
tightness can also be shown 
without imposing that $n=m$ and $m$ is odd. 
\end{proof}

Thus, the {\sc \mecha} rule satisfies at least 
$\frac{n+1}{2}$ agents, indicating that $\rho \ge \frac{n+1}{2n}=\frac{1}{2}+\frac{1}{2n}$ for every number of projects $m$. 
This 
means that one can always satisfy a majority of the population when $\tau=m/2$.  
Now we turn our attention to the upper bound on $\rho$. 

\begin{theorem} \label{upper_bound} 
When $\tau=m/2$, there exists an instance with $m=3$ projects for which no 
solution  
satisfies more than $2n/3+1$ agents.  
\end{theorem}
\begin{proof}
Consider an instance $I$ with $n=3k$ agents for some positive integer $k$. There are $k$ distinct reals $\delta_1, \ldots, \delta_k$ such that 
$0< \delta_1 < \delta_2 < \cdots < \delta_k <1/6$. 
For each $\delta_i$ the instance comprises $3$ agents who have made the following demands.  
\begin{equation} \label{nocover0}
\left( \begin{array}{ccc}
0.5+\delta_i/2 & 0.5-\delta_i& \delta_i/2 \\
\delta_i/2 & 0.5+\delta_i/2  & 0.5-\delta_i \\
0.5-\delta_i & \delta_i/2 & 0.5+\delta_i/2  
\end{array} \right)
\end{equation}
The sub-instance with the $3$ agents associated with $\delta_i$ is called $I_i$. For every pair of reals $(\delta_i,\delta_{i'})$ such that $0< \delta_i < \delta_{i'} <1/6$, we have that 
\begin{equation}
0.5+\delta_{i'}/2 \ge  0.5+\delta_{i}/2 > 0.5-\delta_i \ge   
0.5-\delta_{i'} > \delta_{i'}/2 \ge \delta_{i}/2.
\end{equation}
The solution $(0.5-\delta_k , 0.5+\delta_k/2 , \delta_k/2)$ satisfies the 3 agents of $I_k$, and 2 agents of $I_i$, for each $i\in [k-1]$. Thus, a solution can satisfy $3+2(k-1)=2k+1=2n/3+1$ agents of $I$.  

Let us show now that no solution can satisfy strictly more than $2k+1$ agents of $I$. To do so, we use the following property. 
\begin{center}{\em
``For every pair of distinct reals $(\delta_i,\delta_j)$ such that $0< \delta_i<\delta_j<\frac{1}{6}$, it is impossible to satisfy the 6 agents of $I_i \cup I_j$.''}
\end{center}
Indeed, satisfying at least $2k+2$ agents of $I$ with a single solution imposes that at least two sub-instances $I_i$ and $I_j$ have all their agents satisfied, contradicting the property. 

Let us prove the property by contradiction. Suppose that a solution ${\bf x}=(x_1,x_2,x_3)$ satisfies the 6 agents of $I_i \cup I_j$. If the largest coordinate of ${\bf x}$ is strictly smaller than $0.5+\delta_{j}/2$, then each coordinate of ${\bf x}$ must be at least $0.5-\delta_{j}$ in order to satisfy the 3 agents of $I_j$. This is infeasible because the sum of the coordinates would exceed 1 (indeed, $1/6 > \delta_j \Rightarrow 3(0.5-\delta_j)> 1$). Therefore, at least one coordinate, say $x_1$ w.l.o.g., is at least $0.5+\delta_{j}/2$, and every agent of $I_i \cup I_j$ is locally satisfied for project 1. It remains to locally satisfy at least one project of every agent of $I_i \cup I_j$. The second agent of $I_i$ has declared $0.5+\delta_i/2$ and $0.5-\delta_i$ for projects 2 and 3, respectively. This means that the value of $x_2$ or $x_3$ is at least $\min(0.5+\delta_i/2,0.5-\delta_i)=0.5-\delta_i$. Since $x_1+x_2+x_3 \le 1$ and $\delta_j>\delta_i$, the value of the last coordinate is at most $1-(0.5+\delta_j/2)-(0.5-\delta_i)=\delta_i-\delta_j/2<\delta_i/2$. Because $\delta_i/2$ is the smallest value declared by an agent of $I_i \cup I_j$, we get that the last coordinate does not contribute to the local satisfaction of any project. Therefore, the only possibility is to satisfy the 6 agents of $I_i \cup I_j$ by locally satisfying them for projects $1$ and $2$, or projects 1 and 3, but this is incompatible with $x_1+x_2+x_3 \le 1$. \end{proof}

Closing the gap between $\frac{1}{2}+\frac{1}{2n}$ and $\frac{2}{3}+\frac{1}{n}$ for determining the best ratio $\rho$ is an intriguing open problem. 

\subsection{Extreme Values of the Threshold}

If $\tau=1$, then the following result immediately implies that 
$\rho=1$. 
\begin{proposition} \label{prop:tau=1}
When $\tau=1$, any arbitrary solution $\mathbf{x}$ such that $\sum_{j \in [m]}x_j = 1$ satisfies all agents.
\end{proposition}
\begin{proof} Consider an agent and, without loss of generality, assume that for all $j \in [m-1]$, $\dell_j > x_j$. Then this agent is locally satisfied by the $m^{th}$ project, as $\dell_m \leq 1 - \sum_{j \in [m-1]} \dell_j < 1- \sum_{j \in [m-1]} x_j = x_m$. 
\end{proof}

When $\tau=m-1$, we can suppose that $m > 2$ because when $m=2$, $\tau=1 \Leftrightarrow \tau=m-1$, and $\rho=1$.  

Consider an instance with $n$ agents and $m=3n$ projects. For each $i \in [n]$, agent $i$'s demands for projects $3i-2$, $3i-1$ and $3i$ are equal to $1/3$, and zero everywhere else.
For this instance, Observation \ref{obs2} indicates that we can restrict the coordinates of a solution ${\bf x}$ to $\{0,\frac{1}{3}\}$. In other words, ${\bf x}$ has exactly three coordinates equal to $1/3$, and zeros everywhere else. Such a solution satisfies at most one agent out of $n$. Thus, $\rho \le \frac{1}{n}$. At least one agent is satisfied under the {\sc \mecha} rule so $\rho \ge 1/n$ for all $\tau$.



Observe that $1/n$ is also the best possible value of $\rho$ when $\tau=m$. Indeed, consider an instance containing two agents with demand vectors ${\bf d}$ and ${\bf d}'$ such that $\sum_{j=1}^m \dell_j=\sum_{j=1}^m \dell'_j=1$, and ${\bf d} \neq {\bf d}'$. The only possibility to satisfy ${\bf d}$ (resp., ${\bf d}'$) when $\tau=m$ is to output ${\bf d}$ (resp., ${\bf d}'$). Thus, unless all agents have the same demands, only one agent can be satisfied.

\section{Classes of Instances Where All Agents Can Be Satisfied} \label{sec:sat_all_agents}

This section deals with a characterization, according to $n$ and $m$,  of classes of instances admitting a budget-feasible solution ${\bf x}$ which satisfies all agents. The situation depends on $\tau$, and the most interesting cases are $\tau \in \{m/2,m-1\}$.

\subsection{At-Least-Half Scenario}

\begin{table}
\begin{center}
\caption{\label{tab_HALF}Satisfying all agents when $\tau=m/2$}
\begin{tabular}{l|ccc}
&$m=2$&$m=3$&$m \ge 4$\\
\hline
$n \le 3$&\cmark (Prop.  \ref{prop:tau=1})&\cmark (Thm.  \ref{thm3})&\cmark (Thm.  \ref{thm3})\\
$n \ge 4$&\cmark (Prop.  \ref{prop:tau=1})&\xmark (Exam. \ref{instance_nocover})&\\
\end{tabular}
\end{center}
\end{table}

When $\tau=m/2$, we prove that every instance having at most 3 agents admits a 
solution satisfying everybody, for any number of projects. This also holds for every instance having at most 2 projects, for any number of agents. These results cannot be extended to more agents, or more projects, because of the impossibility to satisfy all agents of Example~\ref{instance_nocover} which contains 3 projects and 4 agents. The results are summarized in Table \ref{tab_HALF}.

When $m=2$ and $\tau=m/2$ (hence $\tau=1$),  we get from Proposition \ref{prop:tau=1} that any solution $(x_1,x_2)$ such that $x_1+x_2=1$ satisfies all agents. 
When $n=2$, we know from Lemma \ref{lem_pair} that the {\sc \mecha} rule always satisfies the two agents for any number of projects. 

Now we move on to the 
most technical result of this section, dealing with 3 agent instances.

\begin{theorem} \label{thm3} When $\tau=m/2$, instances with 3 agents and $m \ge 3$ projects always admit a solution which satisfies all its agents, and such a solution can be computed in polynomial time.    
\end{theorem}

\begin{proof} [Sketched proof] 
We create 3 solutions, namely a $3 \times m$ matrix $\mathcal{S}$, by reshuffling, column by column, and in a non-trivial way, the $3 \times m$ matrix $\D$ containing the agents' demands. The key point is that we can always determine $\mathcal{S}$ (and efficiently compute it) such that each of its lines is a solution satisfying the three 
agents. 
Some lines of $\mathcal{S}$ may violate (\ref{mass_l}), but at least one of them  satisfies it because $\mathcal{S}$ is a permutation of $\D$, and each line of $\D$ verifies (\ref{mass_l}). Therefore, after the careful determination of $\mathcal{S}$, one of its lines satisfying (\ref{mass_l}) is output. 

The full proof appears in the appendix (Section \ref{sec:proof:lem:pair}). 
\end{proof}

Let us illustrate Theorem \ref{thm3} on the following example  
$\D$ with 3 agents and 5 projects. 
\begin{equation} \label{mat1}
\left( \begin{array}{ccccc}
0.3 & 0.33 & 0.04 & 0.16& 0.17\\
0.2 & 0.18 & 0.28& 0.32& 0.02 \\
0 & 0.19 & 0.21& 0.31& 0.29   
\end{array} \right) \end{equation}

Following the full proof of Theorem \ref{thm3},  
matrix (\ref{mat1}) is reshuffled, column by column, to give the following matrix $\mathcal{S}$. 
\begin{equation} \label{mat2}
\left( \begin{array}{ccccc}
0.3 & 0.18 & 0.21 & 0.31& 0.02\\
0 & 0.33 & 0.28& 0.32& 0.17 \\
0.2 & 0.19 & 0.04& 0.16& 0.29   
\end{array} \right) \end{equation}

Each line of (\ref{mat2}) 
satisfies the 3 agents but only the third one  $(0.2 , 0.19 , 0.04, 0.16, 0.29)$ 
is feasible because its total sum does not exceed 1.

\subsection{All-But-One Scenario} \label{sec:saa_ABO}

\begin{table}
\begin{center}
\caption{\label{tab_ABO}Satisfying all agents when $\tau=m-1$}
\begin{tabular}{l|cccc}
&$m=2$&$m=3$&$m=4$&$m=5$\\
\hline
$n=2$&\cmark (Prop.  \ref{prop:tau=1})&\cmark (Thm.  \ref{thm3})&\cmark (Prop. \ref{counter_2_4})&\xmark (Exam. \ref{nocover_five})\\
$n=3$&\cmark (Prop.  \ref{prop:tau=1})&\cmark (Thm.  \ref{thm3})&\xmark (Exam. \ref{nocover_four})&\\
$n=4$&\cmark (Prop.  \ref{prop:tau=1})&\xmark (Exam. \ref{instance_nocover})&&\\
\end{tabular}
\end{center}
\end{table}

Let us identify the instances for which a solution can satisfy all agents when $\tau=m-1$. Our results are summarized in Table \ref{tab_ABO}.  
Using Proposition  \ref{prop:tau=1}, any vector $(x_1,x_2)$ such that $x_1+x_2=1$ satisfies any agent when $m = 2$. 
We know from Theorem  \ref{thm3} that instances with 3 projects and (at most) 3 agents always admit a solution which satisfies all its agents ($\lceil m/2\rceil=m-1$ holds when $m=3$). This positive result cannot be extended to $n\ge 4$ because Example \ref{instance_nocover}, which has  
4 demand vectors\footnote{Multiple agents can have the same demand vector.} and 3 projects, does not admit any solution satisfying all the demands. Similarly, extending the positive result to a larger number of projects is not possible because Example  \ref{nocover_four}, with 3 agents and 4 projects, does not admit any solution satisfying all agents. 
\begin{instance} \label{nocover_four}
$$\left( \begin{array}{cccc}
0.1 & 0.3 & 0.4 & 0.2\\
0.2 & 0.4 & 0.1 & 0.3 \\
0.4 & 0.2 & 0.3 & 0.1 
\end{array} \right)$$
\end{instance}

The proof that one cannot satisfy all agents in Example \ref{nocover_four}, when $\tau=m-1$, is presented in the appendix (Section \ref{sec:sm:1}). Also, note that Example \ref{nocover_four} can be extended to a larger number of agents and projects.

\begin{proposition} \label{counter_2_4} When $\tau=m-1$, 
every instance with 2 agents and $4$ projects admits a solution which satisfies all agents.
\end{proposition}
\begin{proof} Suppose the instance is 
$$\left( \begin{array}{cccc}
a_1 & a_2 & a_3 & a_4\\
b_1 & b_2 & b_3 & b_4
\end{array} \right)$$
where $\sum_{j=1}^4 a_i \le 1$ and  $\sum_{j=1}^4 b_i \le 1$. If ${\bf a}$ (resp., ${\bf b}$) satisfies both agents, then return it. 
Otherwise, neither ${\bf a}$ nor ${\bf b}$ satisfies both agents. This means that for two projects, the demand of the first agent exceeds the demand of the second agent, and it is the other way around for the other two projects. Let us suppose w.l.o.g. that $a_1>b_1$, $a_2>b_2$, $a_3<b_3$, and $a_4<b_4$. Each line of the following array is a solution satisfying both agents,
$$\left( \begin{array}{cccc}
a_1 & b_2 & b_3 & a_4\\
b_1 & a_2 & a_3 & b_4
\end{array} \right)$$
and at least one of them respects the budget constraint (\ref{mass_l}), otherwise contradicting $\sum_{j=1}^4 a_i \leq 1$ and $\sum_{j=1}^4 b_i \leq 1$.
\end{proof}
Unfortunately, Proposition \ref{counter_2_4} does not extend to instances with 2 agents and 5 projects, because of Example \ref{nocover_five} (see the justification in the appendix,  Section \ref{sec:sm:1}).

\begin{instance}  \label{nocover_five}
$$\left( \begin{array}{ccccc}
0.32 & 0.32 & 0.32 & 0.02 & 0.02\\
0.05 & 0.05 & 0.05 & 0.425 & 0.425
\end{array} \right)$$ 
\end{instance}

\subsection{Extreme Values of the Threshold} \label{sec:extreme:classes}
When $\tau=1$, all agents 
can be satisfied, 
for every instance (cf. Proposition \ref{prop:tau=1}). When $\tau=m$, it is not difficult to construct instances with at least $2$ agents and at least $2$ projects where it is impossible to satisfy all agents.

\section{Satisfying All Agents of a Single Instance} 

\label{sec:social_welfare_max}

This section deals with the complexity of the following problem. 
\begin{definition}[\AAS$(\tau)$] Given $\tau$ and the agents' demands, decide whether a budget-feasible solution satisfying all agents exists.
\end{definition}

Note that the problem is easy when $\tau \in \{1,m\}$.  Proposition \ref{prop:tau=1} covers the case 
$\tau=1$. When $\tau=m$, all agents are satisfied by  
the solution $x_j=\max_{i \in N} \dell_j^i$ for all $j \in [m]$, therefore the answer to the above problem is yes if, and only if $\sum_{j\in[m]} ( \max_{i \in N} \dell_j^i ) \le 1$.

We therefore focus on the case $\tau=m-1$ and begin by observing that, unlike Sections \ref{sec:fraction} and \ref{sec:sat_all_agents}, here we \emph{cannot} assume without loss of generality that the demands of each agent sum up to exactly $1$. Indeed, even though it may intuitively seem that allowing agents to express demands that do not use up all the budget should make instances ``easier'' (in the sense that it becomes more likely that a good solution exists), such instances do not necessarily become easier in terms of computational complexity, as this generalizes the class of potential inputs. 

Somewhat surprisingly, we discover that this subtle observation is crucial. Our first result is to show that if we allow instances  where $\sum_{j\in [m]} \dell^i_j <1$, \AAS$(m-1)$ is strongly NP-complete (Theorem \ref{thm:mminus1h}). However, if we impose the restriction that $\sum_{j\in [m]} \dell^i_j =1$, that is, all agents must express demands that use up the budget, then the problem admits a pseudopolynomial time algorithm, and hence cannot be strongly NP-complete, unless P=NP (Theorem \ref{thm:mminus1p}). Interestingly, this subtle distinction only plays a role if $\tau=m-1$: we show that for all $c\ge 2$, \AAS$(m-c)$ is strongly NP-complete, even if all agents have demands that sum up to exactly $1$ (Theorem \ref{thm:mminus2}). 

We leave it as a challenging open  question to determine the complexity of \AAS$(m/2)$.

\subsection{All-But-One Scenario: General Demands}

We begin by showing that \AAS$(m-1)$ is strongly NP-complete when  $\sum_{j\in [m]} \dell^i_j \le 1$ for all $i\in N$ (general demands).

\begin{theorem}\label{thm:mminus1h} 
\AAS$(m-1)$ is strongly NP-complete.  
\end{theorem}

\begin{proof}

We perform a reduction from a \textsc{vertex cover} instance $G=(V,E)$, with target vertex cover size $k$. We
construct one agent for each edge $e\in E$ and one project for each
vertex $v\in V$. If the $i$-th agent represents edge $e$ with endpoints represented by projects $j_1,j_2$, we set $d^i_{j_1}=d^i_{j_2}=\frac{1}{k}$, while all other demands are $0$. Note that for each agent the sum of all demands is $\frac{2}{k}$, which is significantly smaller than $1$ (we can assume that $k$ is large, otherwise the original instance is easy). 

Now, if there is a vertex cover $S$ of size $k$ in $G$, we assign value $\frac{1}{k}$ to the projects representing vertices of $S$ and $0$ to all other projects. We see that for each agent we satisfy at least one of the two strictly positive demands, so we satisfy at least $m-1$ demands of each agent. For the converse direction, if there is a solution satisfying $m-1$ demands of each agent, this solution must give value $\frac{1}{k}$ to at most $k$ projects, from which we extract a set $S$ of at most $k$ vertices. This set must be a vertex cover, because if some edge $e$ is not covered by $S$ the corresponding agent would have $2$ unsatisfied demands.
\end{proof}

\subsection{All-But-One Scenario: Tight Demands}

Theorem \ref{thm:mminus1h} established that \AAS$(m-1)$ is strongly NP-hard, but while crucially relying on agents whose demands are strictly smaller than the budget. We now show that if we only consider agents whose demands exhaust the budget, the problem becomes significantly more tractable. The intuition behind our algorithm is the following: suppose that an agent supplies us with a set of demands that use up the budget and we want to satisfy $m-1$ of her demands. We partition her demands into two sets, so that the sum of the demands of the two sets are as even as possible. In one of the two sets we have to satisfy all demands, so if both sets use up a constant fraction of the budget we can branch, making significant progress at each step (i.e., reducing the available budget by a constant factor). If it is impossible to partition an agent's demands in such a balanced way, we observe that the agent must essentially be placing almost all her budget demand on a single project. However, the case where all agents behave in this single-minded way is easy to handle by guessing at most one project on which we assign a big part of the budget. In this case, the solution is forced for all almost single-minded agents who do not agree with this choice (we \emph{must} satisfy all their other requests) and we can recurse on a much smaller instance.   

\begin{theorem}\label{thm:mminus1p} There is a pseudopolynomial time algorithm
that decides \AAS$(m-1)$  
under the condition that each agent has total demand exactly $1$. 
\end{theorem}

\begin{proof} 

We will describe a recursive algorithm that takes as input an instance of the
problem as well as an initial base solution ${\bf x}$. If ${\bf x}$ already satisfies all but one demands of each agent, the algorithm exits. Otherwise the algorithm will either return a solution ${\bf x^*}$ which dominates ${\bf x}$, that is, for
all $j\in[m]$ we have $x^*_j\ge x_j$, while $\sum_j x^*_j\le 1$ (with the goal of satisfying more demands), or report
that no such ${\bf x^*}$ exists (while using only values corresponding to demands cf.~ Observation~\ref{obs2}). Initially, we set ${\bf x}$ to be the vector
that contains $0$ everywhere.

Given ${\bf x}$, we will define the residual budget $b_r({\bf x})$ (or simply
$b_r$ if ${\bf x}$ is clear from the context) as $b_r({\bf x}) = 1-\sum_j x_j$.
The residual budget is the amount that we are still allowed to use to increase
${\bf x}$ without surpassing the total budget. The high-level strategy of our
algorithm will be to branch in a way that decreases the residual budget by a
constant factor in each step.

We define the residual demand of agent $i$ for project $j$ as $d_r(i,j,{\bf x})
= \max\{\dell^i_j - x_j,0\}$ and write simply $d_r(i,j)$ if ${\bf x}$ is clear
from the context. Abusing notation slightly we will write $d_r(i,J)$ for
$J\subseteq [m]$ to mean $d_r(i,J)=\sum_{j\in J} d_r(i,j)$.

One of the key intuitions of our algorithm is that as long as an agent $i$ is
not satisfied by $\bf x$, we can assume that $d_r(i,[m])\ge b_r$, that is, the
amount needed to fully satisfy $i$ is at least as high as the total remaining
budget. This is shown in the following observation.

\begin{observation}\label{obs:3} For all agents $i$ and initial vectors $\bf x$
we have $\sum_{j\in [m]} d_r(i,j,{\bf x}) \ge b_r({\bf x})$. \end{observation}

\begin{proof} Let $J\subseteq [m]$ be the set of indices such that $\dell^i_j>
x_j$, that is, the set of indices such that $d_r(i,j)>0$. We have $\sum_{j\in
[m]} d_r(i,j) = \sum_{j\in J} d_r(i,j) = \sum_{j\in J} (\dell^i_j - x_j) \ge
\sum_{j\in J} (\dell^i_j - x_j) + \sum_{j\not\in J} (\dell^i_j - x_j) = 1 -
\sum_{j\in [m]} x_j = b_r$.  \end{proof}

Note that a key part of the previous observation is that $\sum_j \dell^i_j=1$;
this is to be expected as Theorem \ref{thm:mminus1h} showed that without this
fact the problem is strongly NP-hard.

We begin our algorithm by formulating some easy base cases. First, suppose
there exists an agent $i$ such that there exists at most one project $j$ with
$\dell^i_j>x_j$, that is, an agent who is already satisfied for all but at most
$1$ project. It is not hard to see that we can safely remove this agent from
the instance, as we intend to return a solution ${\bf x^*}$ that dominates $\bf
x$. If there are no agents remaining, then we return $\bf x$ as our solution. 
Second, suppose that $b_r({\bf x})$ is too small to be of help. More precisely,
suppose that for all agents $i$ and project $j$ for which $\dell^i_j>x_j$ we
also have $\dell^i_j>x_j+b_r$. Then, we report that no solution ${\bf x}^*$ that dominates ${\bf x}$ exists, since
even if we use all our remaining budget on any project, this will not increase
the satisfaction of any agent.

Let us now proceed to the general cases. Consider an agent $i$ who currently
has a set $J$ of unsatisfied projects that contains at least two projects. We
observe that for any $J_1\subseteq J$, any feasible solution $\bf x^*$ must
have one of the following: $\bf x^*$ satisfies all demands of agent $i$ for the
set $J_1$ or $\bf x^*$ satisfies all demands of agent $i$ for $J\setminus J_1$.
This is not hard to see, because $\bf x^*$, if it is feasible, must leave at
most $1$ project unsatisfied for $i$. Our first goal then will be to attempt to
find agents $i$ so that both $J_1$ and $J\setminus J_1$ have large residual
demands for agent $i$ in the sense that increasing $\bf x^*$ to satisfy the
demands of $i$ for either set will spend a constant fraction of $b_r$.

The algorithm now proceeds as follows: for each agent $i$, let $J$ be the set
of projects with $\dell^i_j>x_j$. We produce a greedy partition of $J$ by
sorting the elements of $J$ in non-increasing order of $d_r(i,j)$, and
alternately placing elements in $J_1$ and $J\setminus J_1$. If we have
$\sum_{j\in J_1} d_r(i,j) \ge \frac{b_r}{4}$ and $\sum_{j\in J\setminus J_1}
d_r(i,j) \ge \frac{b_r}{4}$ we say that the partition is \emph{good}.

If the algorithm finds an agent $i$ such that the above procedure gives a good
partition $J_1, J\setminus J_1$, we branch into two instances: in the first
instance we set the new initial vector ${\bf x^1}$ by keeping the same values
as $\bf x$ for $j\not\in J_1$ and setting ${\bf x^1}=\dell^i_j$ for $j\in J_1$;
in the second instance we keep the values of $\bf x$ for $j\not\in J\setminus
J_1$ and set ${\bf x^2}_j = \dell^i_j$ for $j\in J\setminus J_1$. 

Suppose then that the above procedure does not produce a good partition for any
agent. We claim that in this case, for each agent $i$, there exists a unique
project $j$ with $d_r(i,j)>\frac{b_r}{2}$.

\begin{observation} \label{obser_proc_fails} If the 
procedure fails to produce a good partition
for agent $i$, then there exists a unique project $j$ with
$d_r(i,j)>\frac{b_r}{2}$.\end{observation}

\begin{proof} For contradiction, suppose that an agent has either at least two
projects with $d_r(i,j)>\frac{b_r}{2}$ or none. In the first case, the greedy
partition procedure will clearly place at least one project with
$d_r(i,j)>\frac{b_r}{2}$ in both $J_1$ and $J\setminus J_1$, so the partition
is good. In the latter case, we have $\sum_{j\in J_1} d_r(i,j) \ge
\sum_{j\not\in J_1} d_r(i,j) \ge \sum_{j\in J_1} d_r(i,j) - \frac{b_r}{2}$,
where the first inequality is due to the non-increasing order and the second
inequality is because the first project in the ordering has by assumption
$d_r(i,j)\le \frac{b_r}{2}$. Since by Observation \ref{obs:3} we have
$\sum_{j\in J} d_r(i,j) \ge b_r$ we get $\sum_{j\in J_1} d_r(i,j) \ge
\frac{b_r}{2}$ and $2\sum_{j\not\in J_1} d_r(i,j)\ge \frac{b_r}{2}$.  \end{proof}

We will say that a pair of an agent $i$ and project $j$ with
$d_r(i,j)>\frac{b_r}{2}$ is \emph{expensive}. It is clear that any feasible
solution may satisfy at most one expensive pair without exceeding the budget, or more precisely that  any 
solution satisfying all the agents and respecting the budget may satisfy at most one project associated with an expensive pair.

Our algorithm therefore considers the following cases: (i) we decide to satisfy
no expensive pair, that is, for each agent $i$ we increase $\bf x$ so that we
satisfy all the demands of $i$ except the expensive one (ii) we decide that if
an expensive pair is to be satisfied, it will involve some project 
$j$ (we
consider all such possible projects); we then treat all agents for which the
expensive project is not $j$ as in the previous case. For case (i) we have
fully specified a solution, so if it is feasible we return that. For case (ii),
we produce a new instance where the only remaining agents are those for which
the expensive project is the project $j$; we modify ${\bf x}$ by ensuring that the least expensive of the expensive pairs involving project $j$ is satisfied (spending at least $\frac{b_r}{2}$ in the process) and we call the same algorithm
recursively for each such project.

Correctness is not hard to establish for the above procedure as our algorithm
exhaustively considers all possible cases in its branching. 

Let $T(n,b)$ be a function that describes the number
of sub-instances that will be produced from an instance with $n$ agents and
residual budget $b$. For the case where we have a good partition we have
$T(n,b)\le 2T(n,\frac{3b}{4})$. For the case of expensive pairs, we produce
several instances: one for each considered expensive project $j$. It is crucial
now to notice that this produces a partition of the agents into disjoint sets,
because if an agent $i$ had two expensive projects, we would have been able to
produce a good partition for this agent. So, if the number of agents in the
produced instances is $n_1,n_2,\ldots,n_t$, we have $\sum_{i\in [t]} n_i\le n$.
As a result, in this case we have $T(n,b)\le \sum_{i\in [t]} T(n_i,b/2)$. Solving
this recurrence we get $T(n,b)\le n\cdot b^{\log_{4/3} 2}$. We now note that we
only need to run this recursion until the budget becomes low enough to reach
the base case, that is, until the budget becomes smaller than the difference
between any two elements of the input matrix, which gives the pseudopolynomial
time bound. \end{proof}

\subsection{All-But-Two Scenario}

The positive result of Theorem \ref{thm:mminus1p} may tempt one to think that imposing the restriction that all demands sum up to $1$ may lead to (pseudopolynomial) tractability not only when $\tau=m-1$, but also for slightly smaller values of $\tau$. In this section we show that this is not the case by establishing that \AAS$(m-c)$ is strongly NP-complete for all fixed $c\ge 2$ and hence Theorem \ref{thm:mminus1p} represents an island of tractability. 

Let us first give some informal explanation of this result.

A natural way to try to extend the algorithm of Theorem \ref{thm:mminus1p} to the case $\tau=m-2$ would be to evenly partition the demands of an agent into three sets, knowing that the demands of one set must be satisfied. This would still work well for agents with balanced demands; however, things become significantly more complicated when agents present us with lopsided demands. In particular, consider the case of an almost single-minded agent, who requests almost all the budget for one project and small positive values for a few other projects. In the context of Theorem \ref{thm:mminus1p}, such agents are easy to handle, because for most of them we will not be able to satisfy their large request, meaning that we \emph{must} satisfy all their other demands. This is no longer true if $\tau=m-2$, as we still have the margin to reject one other demand of such an agent. Indeed, the case of almost single-minded agents closely resembles the case of Theorem \ref{thm:mminus1h}: we have agents where we need to satisfy all-but-one of their (remaining) demands, and their total demands are significantly smaller than the total budget. As a result, the basic idea of the following theorem is to take the reduction of Theorem \ref{thm:mminus1h} and augment it with a few $(c-1)$ new distinct projects for each agent, to which the agent assigns almost all her budget, but which no other agent cares about. By showing that an optimal solution would ignore these ``private'' projects, we essentially fall back to the same instances of Theorem \ref{thm:mminus1h}.

\begin{theorem}\label{thm:mminus2} 
For all fixed $c\ge 2$, \AAS$(m-c)$ is strongly NP-complete, even if all agents make a total demand of weight exactly $1$.
\end{theorem}

\begin{proof} We present a reduction from \textsc{vertex cover}. We are given a graph
$G=(V,E)$ and an integer $k$ and are asked if $G$ has a vertex cover of size at
most $k$. We assume without loss of generality that $k> c+1$. This can be achieved by taking the disjoint
union of $G$ with $c+1$ disjoint edges and increasing $k$ by $c+1$.

We construct an instance of our problem with $n=|E|$ agents and $m=(c-1)|E|+|V|$
projects. Suppose the edges of $E$ are numbered $e_1,e_2,\ldots,e_{|E|}$ and
the vertices are $v_1,v_2,\ldots,v_{|V|}$. Then, agent $i$ should be thought of as
representing the edge $e_i$. If $e_i=v_{i_1}v_{i_2}$, then we set the demands of the
$i$-th agent as follows: $\dell_{i_1}^i=\dell_{i_2}^i=\frac{1}{k}$; 
$\dell_{|V|+(i-1)(c-1)+1}^i=\dell_{|V|+(i-1)(c-1)+2}^i=\ldots=\dell_{|V|+i(c-1)}^i=\frac{k-2}{(c-1)k}$; and all other demands are $0$. The way to think
about this is that the first $|V|$ projects represent the vertices of $G$ and
the agent representing edge $e_i$ will assign positive demand among them only to the projects
corresponding to the endpoints of $e_i$. The remaining projects are set up in
such a way that each agent has a high demand for exactly $c-1$ distinct projects, which are only interesting to her,
and is indifferent to all the others.

If there is a vertex cover $S$ of size $k$ in $G$, then we assign for each $v_j\in S$
value $\frac{1}{k}$ to $x_j$ and value $0$ everywhere else. Consider now an
agent representing edge $e_i$. Since this agent has only $c+1$ strictly
positive demands, it is sufficient to show that one of them is satisfied.
This must be the case, because $S$ is a vertex cover, therefore contains at
least one endpoint of $e_i$, and we have assigned a value to the corresponding
project that covers the demand of the $i$-th agent. For the converse direction, suppose there is an assignment that satisfies at least $m-c$ requests of each agent, and among such assignments pick one with minimum total value. We observe that such an assignment must assign positive value only to some of the first $|V|$ projects. Indeed, otherwise we would assign value $\frac{k-2}{(c-1)k}$ to a project for which only a single agent has positive demand. By setting the value assigned to this project to $0$ we do not affect any other agent, while if this agent now has $c+1$ unsatisfied requests, we can assign value $\frac{1}{k}$ to one of her other requests. We recall that $k>c+1$ which implies that $k-2>c-1 \Rightarrow \frac{k-2}{(c-1)k} > \frac{1}{k}$, so the new solution has strictly lower total value than the original one, contradicting the original selection.
We therefore have that a solution to the new instance must assign
positive weight, thus weight $\frac{1}{k}$ (cf. Observation \ref{obs2}), only to projects representing
vertices of $G$. Each agent then already has $c-1$ non-satisfied project (the ones
with demand $\frac{k-2}{(c-1)k}$). It must then be the case that the projects
selected form a vertex cover, as if there were an uncovered edge, the
corresponding agent would have $c+1$ unsatisfied projects.  \end{proof}

\section{Minimizing the Total Budget Spent} \label{sec:min_tot_bud}

Previous sections  
deal with problems 
where the total budget is at most $1$, and we  
have
seen that in many situations it can become challenging to satisfy all agents.
We now therefore consider a natural related optimization question: assuming
that the demands of each agent sum up to $1$, corresponding to the most demanding configuration imposed by (\ref{mass_l}), what is the minimum total budget that needs to be spent to satisfy everyone? 


It is worth stressing that 
the optimal solution may have value
strictly more than $1$, but also strictly less than $1$. Indeed, when $\tau < m$, it is not hard to construct instances where one can satisfy all agents while using budget that
is $o(1)$, if for example we have agents who assign almost all their demand to
a single project and give very low demands for the others. Hence, from a
computational point of view, the fact that we can always satisfy all agents
with a budget of at most $b$ does not in itself provide a $b$-approximation for
the optimization problem, because we cannot in general lower bound the value of
the optimum.

In this section, we address the minimization of the total budget in two ways, for every $\tau \in \{1,m/2,m-1,m\}$. The first approach consists of determining the complexity of minimizing the total budget of a solution satisfying all agents, when the input is a \emph{single} instance (i.e., an $n \times m$ matrix of demands). The second approach consists of providing, for each $m$, a solution of minimum total budget satisfying all agents for the entire class of instances having $m$ projects (and any number of agents).  
Our results 
are summarized in Table \ref{tab_total_budget}. The first line of the table  indicates that the first approach leads to NP-hard problems, except when $\tau=m$. The second line 
comprises the total budget of our solutions for the second approach, and these values are (asymptotically) tight.

\begin{table}
\begin{center}
\caption{\label{tab_total_budget} Satisfying all agents with minimum total budget}

\begin{tabular}{c|cccc}
$\tau$&$1$&$m/2$&$m-1$&$m$\\
\hline
\small{Single instance complexity}&NP-hard&NP-hard&NP-hard&P\\
\small{(first approach)} & & & & \\
\hline
\small{Total budget upper bound}&1&$2$&$m/2$& $m$\\
\small{(second approach)}&&&&\\
\end{tabular}

\end{center}
\end{table}

\subsection{At-Least-Half Scenario}

We consider that $\tau=m/2$ in this section. 

\begin{theorem} When $\tau=m/2$, computing the minimum budget necessary to satisfy every agent
is NP-hard. \end{theorem}

\begin{proof}

We present a reduction from \textsc{Independent Set}. We are given a graph
$G=(V,E)$ and an integer $k$ and are asked if $G$ has an independent set of
size at least $k$.  To avoid confusion we denote the number of vertices and
edges of $G$ as $\mathscr{N}=|V|$ and $\mathscr{M}=|E|$. 
Assume without loss of generality that
$k<\mathscr{N}/2$ (otherwise we can add universal vertices to $G$; since these vertices
are connected to everything they do not increase the size of the maximum
independent set). 

We construct an instance with $n=\mathscr{M}$ agents, and $m=2\mathscr{N}-2k-2$ projects, that is,
we construct one agent for each edge. The projects can be thought of as being
partitioned into $\mathscr{N}$ ``interesting'' projects and $\mathscr{N}-2k-2$ ``non-interesting''
projects. All agents will assign value $0$ to all non-interesting projects.

To describe the values assigned to the interesting projects, number the
vertices of $G$ with integers from $\{1,\ldots,\mathscr{N}\}$. Each vertex corresponds to
an interesting project. Consider the edge $e_i$ whose endpoints are $i_1,i_2$.
Then, the $i$-th agent of our instance has demand $\frac{1}{\mathscr{N}-2}$ for all
interesting projects, except projects $i_1$ and $i_2$ for which the agent has
demand $0$.

This completes the construction and we claim that it is possible to satisfy all
agents with a budget of $\frac{k}{\mathscr{N}-2}$ if and only if $G$ has an independent
set of size at least $k$. For one direction, given an independent set of size $k$,
we assign value $\frac{1}{\mathscr{N}-2}$ to all interesting projects corresponding to
vertices of the independent set and $0$ to all other projects. We claim that
this satisfies all agents. Indeed, consider an agent representing the edge
$e_i$. This agent is happy with the assignment to the $\mathscr{N}-2k-2$ non-interesting
projects, as well as the two interesting projects corresponding to the
endpoints of $e_i$. Furthermore, this agent is happy with the assignment to at
least $k-1$ interesting projects (distinct from the
endpoints of $e_i$), since the independent set cannot contain both
endpoints of $e_i$. Hence, the agent is happy with the assignment to
$\mathscr{N}-k-1=m/2$ projects and is satisfied by the proposed solution.

For the converse direction, it is not hard to see that the optimal assignment
will give $0$ to all non-interesting projects and value either $\frac{1}{\mathscr{N}-2}$
or $0$ to each interesting project  (cf.\ Observation \ref{obs2}). Furthermore, we can assume without loss of
generality that the best assignment uses the full budget. Therefore, if there
is a satisfying assignment with budget $\frac{k}{\mathscr{N}-2}$, it must have assigned
strictly positive value to exactly $k$ interesting projects. We claim that the
$k$ corresponding vertices are an independent set of $G$.  Indeed, suppose that
$e_i$ has both endpoints in this set. Then, the agent corresponding to $e_i$
would not be satisfied, as the proposed solution would only cover $k$ of this
agent's demands for interesting projects.  \end{proof}

\begin{theorem} When $\tau=m/2$, the solution ${\bf x}$ that sets $x_j=\frac{2}{m}$ for all $j\in[m]$ satisfies all agents. In addition, there exists, for each $m$,  an instance with $m$ projects, such 
that the minimum budget necessary to satisfy all agents is at least
$2-\frac{6}{m}$.
\end{theorem}

\begin{proof}

We observe that if there was an agent not satisfied by ${\bf x}$,  
this agent would have demand strictly more than $1$. Indeed, for such an agent there must be at least $m/2$ projects for which the agent has demand strictly more
than ${2}/{m}$.  



We will prove the lower bound of $2-6/m$ 
for $m$ even. If $m$ is odd, then we can construct
the instance we would have built for $m-1$ and add a project for which all
agents have demand $0$. Since in the smaller instance, any solution with the
claimed budget will have an agent who is unhappy for at least $\frac{m-1}{2}+1
= \frac{m+1}{2}$ projects, this agent is still unsatisfied in the new instance.

We will only consider agents that assign demands which are integer multiples of
$\frac{1}{m^2}$. For each vector  of $m$ non-negative integers $\{d_1,\ldots,d_m\}$
such that $\sum_{j=1}^m d_j = m^2$, construct an agent with index $i$ and demands $\dell^i_j =
\frac{d_j}{m^2}$ for all $j\in[m]$. Note that the instance we are constructing has size
exponential in $m$. 


We now claim that if we have a solution with budget at most $2-\frac{6}{m}$,
then some agent is unsatisfied. To see this, take such a solution ${\bf x}$
and assume without loss of generality that $x_i\le x_{i+1}$ for all $i$, that
is, the projects are sorted in non-decreasing order according to the value
assigned by the optimal solution. We will show that ${\bf x}$ leaves some agent unsatisfied. 


Consider now the following case: if $\sum_{i=1}^{m/2+1} x_i \le 
1-\frac{m/2+1}{m^2}$, then we claim that an unsatisfied agent exists.  Indeed,
without loss of generality the optimal solution sets each $x_i$ to be an
integer multiple of $\frac{1}{m^2}$ for all projects, so there must exist an
agent, say with index $i$, who sets $\dell^i_j=x_j+\frac{1}{m^2}$ for all
$j\in\{1,\ldots,m/2+1\}$.  
Note that we have that $\sum_{j=1}^{m/2+1}\dell^i_j \le
1-\frac{m/2+1}{m^2}+\frac{m/2+1}{m^2} = 1$, hence such an agent exists and is
unsatisfied regardless of the assignment to the remaining projects.

If all agents are satisfied, we therefore must have that $\sum_{i=1}^{m/2+1} x_i > 1-\frac{m/2+1}{m^2}$.
This implies that $\sum_{i=m/2+2}^m x_i \le 2-\frac{6}{m}-1+\frac{m/2+1}{m^2} <
1- \frac{2}{m}$. Because the $x_i$'s are sorted in non-decreasing order we have
that $x_{m/2}+x_{m/2+1}\le 2\cdot \frac{\sum_{i=m/2+2}^m x_i}{m/2-1}\le 4\cdot
\frac{1-\frac{2}{m}}{m-2} \le \frac{4}{m}$.  But then we have that 
$\sum_{i=1}^m x_i = \sum_{i=1}^{m/2+1} x_i + \sum_{i=m/2}^{m} x_i -
(x_{m/2}+x_{m/2+1})$ $\ge 2\left(\sum_{i=1}^{m/2+1} x_i\right) - \frac{4}{m} \ge 2-\frac{6}{m}$
where we have used that since the $x_i$'s are sorted the second sum
is at least as large as the first. We conclude that if a solution satisfies all agents, it must assign sufficient values to the first $\frac{m}{2}+1$ projects to avoid the case of the previous paragraph (which leaves an agent unsatisfied), therefore it must be the case that the total budget used is more than $2-\frac{6}{m}$. \end{proof}

\subsection{All-But-One Scenario}

In this section we suppose that $\tau=m-1$.

\begin{theorem} \label{thm:5.3} When $\tau=m-1$, computing the minimum budget necessary to satisfy every agent is NP-hard.  \end{theorem}

\begin{proof}
We present a reduction from \textsc{vertex cover}. Given a graph $G=(V,E)$ and an integer $k$, we are asked if $V$ has a subset $C$ such that $\{u,v\} \cap C \neq \emptyset$ for every edge $\{u,v\} \in E$, and $|C| \le k$. 

Take an instance of \textsc{vertex cover} and create an instance of the budget division problem where the goal is to satisfy all agents with a minimum total budget.  Every vertex $v_j \in V$ is associated with a project $j$. Every edge $\{v_j,v_{j'}\} \in E$ corresponds to an agent whose demands are $1/2$ for projects $j$ and $j'$, and zero for the other projects. By Observation \ref{obs2}, the coordinates of any solution ${\bf x}$ to this instance of budget are either 0 or $1/2$. Since $\tau=m-1$, an agent associated with $\{v_j,v_{j'}\}$ is satisfied by ${\bf x}$ if $\max \{x_j,x_{j'}\}=1/2$. 

Then, it is not difficult to see that a solution ${\bf x}$ satisfying all agents and having a total budget of $B/2$ corresponds to a vertex cover $C=\{v_j : x_j=1/2\}$ of $G$ of cardinality $B$. \end{proof}

\begin{theorem} When $\tau=m-1$, the solution ${\bf x} $ that sets $x_j=\frac{1}{2}$ for all
$j\in[m]$ has total budget $\frac{m}{2}$ and ${\bf x}$ satisfies all agents. In addition,  there exists, for each $m$,  an instance with $m$ projects, such
that the minimum budget necessary to satisfy all agents is at least
$\frac{m}{2}-\frac{1}{2m}$.  \end{theorem}
\begin{proof}
The solution with $1/2$ on each coordinate satisfies all agents because (\ref{mass_l}) implies that  an agent's demand can exceed $1/2$ for at most one project. 

For the second part,  consider an instance such that, for every pair of distinct projects $(j,j')$, there exists:
\begin{itemize}
\item an agent whose demands for $j$ and $j'$ are equal to $1/2$, and 0 for the other projects, 
\item an agent whose demands for $j$ and $j'$ are equal to $\epsilon$ and  $1-\epsilon$, respectively, and 0 for any other project,  where $\epsilon \in (0,1/2)$. 
\end{itemize}
Let ${\bf x}^*$ be a solution satisfying all the agents when $\tau=m-1$.  
At most one coordinate of ${\bf x^*}$, say $1$ w.l.o.g., can be strictly smaller than $1/2$ because of the agents having two demands of $1/2$. If $x^*_1$ is strictly smaller than $\epsilon$, then every other coordinate must be at least $1-\epsilon$ in order to satisfy the agents having a demand of $\epsilon$ on project $1$, and a demand of $1-\epsilon$ on some project $j'\neq 1$.     

So, the total budget of ${\bf x^*}$ is either  $\epsilon + \frac{m-1}{2}$ or $(m-1)(1-\epsilon)$, corresponding to solutions $(\epsilon,\frac{1}{2},\frac{1}{2},\ldots,\frac{1}{2})$ and $(0,1-\epsilon,1-\epsilon,\ldots,1-\epsilon)$, respectively. Our best lower bound on the total budget of ${\bf x^*}$, i.e., $\min \{\epsilon + \frac{m-1}{2},(m-1)(1-\epsilon)\}$, is obtained when $\epsilon + \frac{m-1}{2}=(m-1)(1-\epsilon) \Leftrightarrow \epsilon=\frac{1}{2}-\frac{1}{2m}$. This gives a lower bound of $\frac{m}{2}-\frac{1}{2m}$.  
\end{proof}

 \subsection{Extreme Values of the Threshold} \label{sec:extreme:minbudget}

Let us first consider the case  $\tau=1$. 

\begin{theorem} \label{m_egal_one} When $\tau=1$, computing the minimum budget necessary to satisfy all agents 
is NP-hard. \end{theorem}

\begin{proof} 
We present a reduction from \textsc{vertex cover}. We are given a graph $G=(V,E)$ and an integer $k$ and are asked if $G$ has a vertex cover of size $k$. We will construct an instance of our problem with one agent for each edge of $E$ and one project for each vertex of $V$, so to simplify notation assume that $n=|E|$ and $m=|V|$.

Suppose that the edges of $G$ are numbered $e_1,\ldots,e_n$ and the vertices are numbered $v_1,\ldots,v_m$. In the new instance agent $i$ will represent the edge $e_i$. If the endpoints of $e_i$ are the vertices $v_{j_1},v_{j_2}$, then we set $d^i_{j_1}=d^i_{j_2}=\frac{1}{m^2}$, while for all other $j\in [m]\setminus\{j_1,j_2\}$ we set $d^i_j=\frac{m^2-2}{(m-2)m^2}$. Observe that since each edge has exactly two endpoints, we have $\sum_{j\in[m]}d^i_j=1$ for all $i\in [n]$. This completes the construction and we claim that $G$ has a vertex cover of size $k$ if, and only if, it is possible to satisfy at least one demand of each agent with a total budget of $\frac{k}{m^2}$.

For one direction, if we have a vertex cover $S$ of size $k$, we assign value $\frac{1}{m^2}$ to each project $j\in[m]$ such that $v_j\in S$ and value $0$ to other projects. We claim this satisfies all agents, since each edge $e_i$ has an endpoint, say $v_j$, that belongs in $S$, therefore agent $i$ has demand $\frac{1}{m^2}$ for project $j$, which was assigned exactly this value.

For the converse direction, we first observe that $\frac{k}{m^2}<\frac{m^2-2}{(m-2)m^2}$ $\Leftrightarrow k<m+2 +\frac{2}{m-2}$, which we can assume without loss of generality, otherwise the original instance is trivial. As a result, a solution that satisfies at least one demand of each agent can be assumed to only use values $\frac{1}{m^2}$ and $0$ (cf. Observation~\ref{obs2}), therefore if the total budget is $\frac{k}{m^2}$ we can construct a vertex cover $S$ by placing $v_j$ in $S$ whenever the solution assigns positive value to project $j$. The set $S$ will be a vertex cover, because if edge $e_i$ were left uncovered this would imply that agent $i$ had all her demands unsatisfied. 
\end{proof}

We know from Proposition \ref{prop:tau=1} that when $\tau=1$, all agents are satisfied by an arbitrary solution of total budget 1. Let us prove that 1 is the minimum budget. Suppose by contradiction that there exists a solution ${\bf y}$ satisfying all agents and $\sum_{j=1}^m y_j <1$, for any possible instance with $m$ projects. Let $\delta=1-\sum_{j=1}^m y_j>0$, and let $(y_1+\frac{\delta}{m},y_2+\frac{\delta}{m},\ldots,y_m+\frac{\delta}{m})$ be the demand vector of some agent. Clearly, ${\bf y}$ does not satisfy the agent's demand.

To conclude this part, consider the case $\tau=m$ where a solution ${\bf x}$ satisfying all agents must verify $x_j \ge \max_{i \in N} \dell_j^i$ for all $j \in [m]$. Thus, the solution $x_j=\max_{i \in N} \dell_j^i$ for all $j \in [m]$ is optimal and the problem is in P. The solution $x_j=1$ for all $j \in [m]$ satisfies all agents, for any instance. Its total budget is $m$, which is optimal if we consider the instance where for every $j\in [m]$, there exists an agent who has a demand of 1 for project $j$, and 0 everywhere else. 

\section{Conclusion and Future Work}

We have considered the problem of existence and computation of a budget division that best satisfies a group of agents. In our setting, an agent is satisfied if she is locally satisfied for at least $\tau$ projects, and local satisfaction of a project $j$ is achieved if the part of the budget allocated to $j$ is not below the agent's demand. We provided a number of results concerning the satisfaction of all agents, or a fraction of them, with a single solution.  

Some questions are left open, such as determining the exact best fraction $\rho$ of agents that can be satisfied in any instance when $\tau=m/2$; we have seen 
that $\rho \in [ 1/2+1/2n,2/3+1/n]$ in this case. A related question is finding the largest value of $\tau$ such that a constant fraction of agents $\rho$ is satisfied. Indeed, since  $\rho = 1/n$   
when $\tau \geq m-1$, there is a phase transition in the interval $[\frac{m}{2},m-1]$.  
Another open problem 
is the determination of the complexity of \AAS$(m/2)$. 

Other related problems can be posed: Given $\tau \in [m]$, how hard is the problem of maximizing the number of satisfied agents with a budget-feasible solution? Solving this problem would nicely complement our results of Section \ref{sec:fraction}, since we would have to deal with a single instance rather than the set of all instances. For each instance considered in isolation, one can go beyond the guarantees established in Section \ref{sec:fraction} regarding the number of satisfied agents. Nevertheless, it is likely that this optimization approach is computationally difficult.

One can also be interested in maximizing a social welfare function: Finding a budget-feasible solution which maximizes either the minimum number of times at an agent is locally satisfied (egalitarian approach), or the total number of pairs $(i,j)$ such that agent $i$ is locally satisfied for project $j$ (utilitarian approach). The egalitarian approach is open whereas the utilitarian approach can be solved in polynomial time (cf. Section \ref{sec:dp:util} in the appendix).     

We have considered deterministic solutions but randomization may help if outputting a probability distribution over several budget feasible solutions is possible \cite{GoyalSSG23}. Then, what would be the best possible expected fraction of satisfied agents? 

In the present work, and as in  \cite{WagnerM23,CCP24,FS24,EGLST26}, each agent 
reports a demand vector. 
The question then arises of the cognitive difficulty of this task, 
in comparison with the cognitively simpler task of approving, or ranking,  projects, as it is considered in \cite{AACKLP23,MPS20,BGPSW24}.

Our final remark is related to the strategic aspect of demands. Since agents may misreport their demands in order to influence the outcome \cite{FPPV21}, 
strategyproof algorithms would be relevant.



\section*{Declarations}
\subsection*{Fundings}

Laurent Gourv{\`e}s is supported  by  Agence  Nationale  de  la  Recherche  (ANR),  project THEMIS ANR-20-CE23-0018. Aris Pagourtzis has been partially supported by project MIS 5154714 of the National Recovery and Resilience Plan Greece 2.0 funded by the European Union under the NextGenerationEU Program. 
Michael Lampis is supported by ANR project S-EX-AP-PE-AL ANR-21-CE48-0022. 
Nikolaos Melissinos is supported by Charles University projects UNCE 24/SCI/008 and PRIMUS 24/SCI/012, and by the project 25-17221S of GAČR. This work received the support of the Greek-French Cooperation Programme
Hubert Curien 2025 -- ARISTOTE.

\subsection*{Author Contributions}

All authors contributed equally to the manuscript.

\subsection*{Data Availability}
No datasets were generated or analysed during the current study.



\subsection*{Conflict of interest} 
The authors have no conflict of interest to declare that are relevant to the content of
this article.

\bibliographystyle{plain} 
\bibliography{sample}

\appendix

\section{Impossibility of Satisfying all Agents} \label{sec:sm:1}

\begin{observation} \label{imposs}
It is impossible to satisfy all the agents of 
Example \ref{instance_nocover} 
when $\tau=2$.
\end{observation}

\begin{instancecont}{instance_nocover}
$$\left( \begin{array}{ccc}
0.5 & 0.5 & 0 \\
0 & 0.5 & 0.5 \\
0.6 & 0.1 & 0.3 \\
0.3 & 0.1 & 0.6 
\end{array} \right)$$
\end{instancecont}

\begin{proof}

In a hypothetical solution $(x_1,x_2,x_3)$ that satisfies every agent of Instance (\ref{instance_nocover}), there must be one coordinate with value at least $0.5$ because of the first two agents. 

If $x_2 \ge 0.5$, then project 2 is satisfied for everybody. In order to satisfy the third agent, the only possibility is to set $x_3 \ge 0.3$ because $x_1 \ge 0.6$ is incompatible with $x_2 \ge 0.5$ and $x_1+x_2+x_3=1$. Since $x_3 \ge 0.6$ is also incompatible with $x_2 \ge 0.5$ and $x_1+x_2+x_3=1$, one would need to set $x_1 \ge 0.3$ in order to satisfy the fourth agent. An impossibility is reached. 

Now suppose $x_1 \ge 0.5$. We need to set $x_2 \ge 0.5$ or $x_3 \ge 0.5$ in order to satisfy agent 2. Due to $x_1+x_2+x_3=1$, the two cases are $x_1 =x_2= 0.5$ and $x_1=x_3= 0.5$, but $x_2=0.5$ has already been considered in the previous paragraph. Thus, we retain $x_1=x_3= 0.5$ and $x_2=0$, which does not satisfy agents 3 and 4.           

Due to the symmetry of Instance (\ref{instance_nocover}), the case $x_3 \ge 0.5$ is the same as the previous one where $x_1 \ge 0.5$. \end{proof}

\begin{observation} \label{imposs2}
It is impossible to satisfy all the agents of 
Example \ref{nocover_four} 
when $\tau=m-1$.
\end{observation}

\begin{instancecont}{nocover_four}
$$\left( \begin{array}{cccc}
0.1 & 0.3 & 0.4 & 0.2\\
0.2 & 0.4 & 0.1 & 0.3 \\
0.4 & 0.2 & 0.3 & 0.1 
\end{array} \right)$$
\end{instancecont}

\begin{proof}  
For contradiction, assume that all agents can be satisfied. First, observe that if, for any fixed project, no agent is locally satisfied, then all agents must be satisfied by the remaining three projects, which is impossible because the budget is 1.

Now we prove that any budget division that locally satisfies exactly one agent for a given project cannot satisfy all agents overall. To see this, assume that for project $j$ only agent $i$ is locally satisfied. Then, the other two agents must be locally satisfied for each of the remaining three projects. Moreover, the choice of $j$ uniquely determines such a budget division.

For example, if $j=1$, then necessarily $i=1$, because agent $1$'s demand for project $1$ is the minimum; it must therefore be $0.2> x_1 \ge 0.1$ (otherwise either none or at least two agents  would be locally satisfied for $j$). Then, in order to satisfy agents $2$ and $3$, we must have $x_2 \ge 0.4$, $x_3 \ge 0.3$, and $x_4 \ge 0.3$, yielding a total budget of at least $1.1$, which is a contradiction. Below we present, for each possible case, the minimum required budget for each project as well as the total required budget:

$$
\begin{array}{c|c|c|c|c}
\text{project } j & 1 & 2 & 3 & 4\\
\text{unique agent } i \text{ locally satisfied by } x_j & 1 & 3 & 2 & 3 \\
\hline
x_1 \ge & 0.1 & 0.2 & 0.4 & 0.2 \\
x_2 \ge & 0.4 & 0.2 & 0.3 & 0.4 \\
x_3 \ge & 0.3 & 0.4 & 0.1 & 0.4 \\
x_4 \ge & 0.3 & 0.3 & 0.2 & 0.1 \\
\text{total} \ge & 1.1 & 1.1 & 1.0 & 1.1
\end{array}
$$

Note that the only case where we do not immediately exceed the budget is when the proposed budget division locally satisfies only one agent in project $3$ (namely, agent $2$). However, this uniquely defines a budget division that does not satisfy agent $2$.

Therefore, the budget assigned to each project must locally satisfy at least two agents. This implies that $x_1, x_4 \ge 0.2$ and $x_2, x_3 \ge 0.3$, and thus ${\bf x} = (0.2, 0.3, 0.3, 0.2)$. However, ${\bf x}$ does not satisfy agent $2$.

Thus, no budget division satisfies all agents when $\tau = m-1$. \end{proof}

\begin{observation} \label{imposs3}
It is impossible to satisfy all the agents of 
Example \ref{nocover_five} 
when $\tau=m-1$.
\end{observation}

\begin{instancecont}{nocover_five}
$$\left( \begin{array}{ccccc}
0.32 & 0.32 & 0.32 & 0.02 & 0.02\\
0.05 & 0.05 & 0.05 & 0.425 & 0.425
\end{array} \right)$$ 
\end{instancecont}
\begin{proof}
To see that we cannot satisfy both agents in Example \ref{nocover_five}, observe that, since $\tau = m-1$, the following conditions must hold for any feasible solution ${\bf x}$:
\begin{itemize}
    \item at least two of the three values $x_1$, $x_2$, and $x_3$ are at least $0.32$, and
    \item at least one of the two values $x_4$ and $x_5$ is at least $0.425$.
\end{itemize}
Therefore, $\sum_{j \in [5]} x_j \ge 1.065$, contradicting the feasibility of ${\bf x}$.
\end{proof}


\section{Proof of Theorem \ref{thm3}}
\label{sec:proof:lem:pair}
The high-level strategy of the proof is given in the sketched proof of Theorem \ref{thm3}.

\begin{proof} The agents are named $a$, $b$, and $c$, and their demands are as follows.
\begin{equation} \label{gmatrix} 
\left( \begin{array}{ccccc}
a_1 & a_2 & a_3 &\cdots&a_m\\
b_1 & b_2 & b_3 &\cdots&b_m\\
c_1 & c_2 & c_3 &\cdots&c_m  
\end{array} \right)
\end{equation}
We assume that  
\begin{equation} \label{constantsum3}
\sum_{j=1}^m a_j=\sum_{j=1}^m b_j =\sum_{j=1}^m c_j=1,
\end{equation}
and if it is not the case, then without loss of generality the values can be arbitrarily increased until (\ref{constantsum3}) holds. We also suppose that the number of projects $m$ is equal to $2k+1$ because of the following observation. 

\begin{observation} When $\tau=m/2$, if an algorithm can output a solution 
which satisfies $\rho \cdot n$ agents for every possible instance with $n$ agents and $m=2k+1$ projects, then the same algorithm can output a solution which satisfies $\rho \cdot n$ agents for every possible instance with $n$ agents and $2k+2$ projects.   
\end{observation}
\begin{proof} Take the instance $I$ with $m=2k+2$ projects and remove one project arbitrarily. We get a feasible instance $I'$ (i.e., (\ref{mass_l}) is satisfied) for which the algorithm outputs a solution ${\bf x}$ which satisfies $\rho \cdot n$ agents. An agent $i$ is satisfied by ${\bf x}$ in $I'$ if she is locally satisfied for at least $k+1$ projects. Thus, agent $i$ is also satisfied by ${\bf x}$ in $I$ because she also needs to be locally satisfied for at least $k+1$ projects. \end{proof}

Using Lemma \ref{lem_pair} we know that $a$ or $b$ $\tau$-covers\footnote{As mentioned in Definition \ref{def_cov},  ${\bf x}$ $\tau$-covers ${\bf y}$ when $x_j \ge y_j$ holds for at least $\tau$ distinct coordinates $j$.}
 both $a$ and $b$. Same remark for the pairs $(a,c)$ and $(b,c)$. 
Thus, if one agent in $\{a,b,c\}$ $\tau$-covers all the agents, then we are done. Otherwise, we can suppose that $a$ $\tau$-covers $a$ and $b$ but not $c$, $b$ $\tau$-covers $b$ and $c$ but not $a$, and $c$ $\tau$-covers $a$ and $c$ but not $b$. This assumption is made w.l.o.g. because the names of the agents can be swapped, if necessary.

We are going to produce three solutions by reshuffling matrix (\ref{gmatrix}) in such a way that every line of the new matrix is a solution which  
$\tau$-covers the 3 agents, and at least one of these solutions is feasible (i.e., the sum of its elements does not exceed 1). Concretely, we are going to swap the elements of every column of (\ref{gmatrix}). It is not difficult to see that these permutations do not prevent the new matrix to have at least one feasible line. Indeed, (\ref{constantsum3}) guarantees that the total sum of the elements is 3, so it cannot be the case that the sum of every line exceeds 1.

Consider the set of indices $S_{ab} \subset [m]$ where $a_j>b_j$. Similarly, $S_{ca}$ contains the indices $j \in [m]$ such that $c_j>a_j$, and $S_{bc}$ contains the indices $j \in [m]$ such that $b_j>c_j$. See Figure \ref{fig1} for an illustration. 
Note that indices $j$ of $[m]$ which are out of $S_{ab}\cup S_{bc}\cup S_{ca}$ satisfy $a_j\le b_j \le c_j \le a_j$, i.e., $a_j=b_j=c_j$. In addition, no index $j$ can be in $S_{ab}\cap S_{bc}\cap S_{ca}$ because it would lead to the contradiction $a_j> b_j> c_j > a_j$.

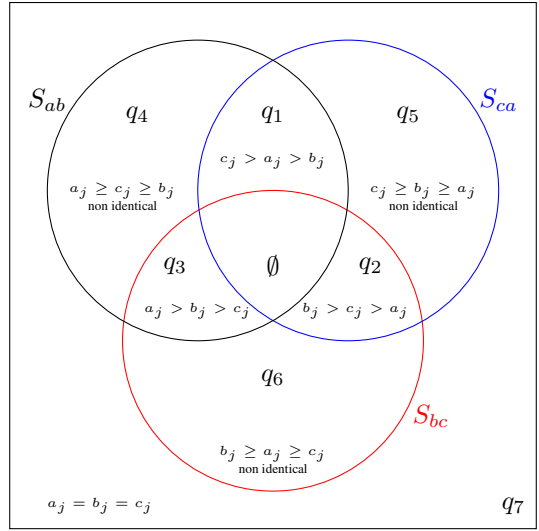
\begin{figure}
\begin{center}
\begin{tikzpicture}[scale=1]

\draw  (-3.5,-2.5) rectangle (3.5,4.5);
\draw[red] (0,0) circle (2) ;
\draw[blue] (1,2) circle (2) ;
\draw (-1,2) circle (2) ;

\draw (-2.3,-2.2) node{{\tiny $a_j=b_j=c_j$}} ;
\draw (-3,3.2) node{$S_{ab}$} ;
\draw[blue] (3,3.2) node{$S_{ca}$} ;
\draw[red] (2.1,-1) node{$S_{bc}$} ;

\draw (-1,0.4) node{{\tiny $a_j>b_j>c_j$}} ;
\draw (1.1,0.4) node{{\tiny $b_j>c_j>a_j$}} ;
\draw (0,2.4) node{{\tiny $c_j>a_j>b_j$}} ;

\draw (-2,2) node{{\tiny $a_j\ge c_j \ge b_j$}} ;
\draw (-2,1.8) node{{\tiny non identical}} ;
\draw (2,2) node{{\tiny $c_j \ge b_j \ge a_j$}} ;
\draw (2,1.8) node{{\tiny non identical}} ;
\draw (0,-1.5) node{{\tiny $b_j \ge a_j \ge c_j$}} ;
\draw (0,-1.7) node{{\tiny non identical}} ;

\draw (0,1) node{$\emptyset$} ;

\draw (0,3) node{$q_1$} ;
\draw (1.3,1) node{$q_2$} ;
\draw (-1.3,1) node{$q_3$} ;
\draw (-1.8,3) node{$q_4$} ;
\draw (1.8,3) node{$q_5$} ;
\draw (0,-0.5) node{$q_6$} ;
\draw (3.2,-2.2) node{$q_7$} ;

\end{tikzpicture}

\end{center}

\caption{\label{fig1} Illustration of Theorem \ref{thm3}}
\end{figure}

Let $q_1$ be the cardinality of the set $Q_1:=S_{ab} \cap S_{ca} \setminus S_{bc}$. We have $c_j>a_j>b_j$ for every index $j \in Q_1$.  
Let $q_2$ be the cardinality of the set $Q_2:=S_{ca} \cap S_{bc} \setminus S_{ab}$. We have $b_j>c_j>a_j$ for every index $j \in Q_2$. 
Let $q_3$ be the cardinality of the set $Q_3:=S_{ab} \cap S_{bc} \setminus S_{ca}$.  We have $a_j>b_j>c_j$ for every index $j \in Q_3$. 
Let $q_4$ be the cardinality of the set $Q_4:=S_{ab} \setminus (S_{ca} \cup S_{bc})$. 
We have $a_j \ge c_j \ge b_j$ with at least one strict inequality for every index $j \in Q_4$ (i.e., the case $a_j=c_j=b_j$ is excluded, so {\em non-identical} is written on Figure \ref{fig1} to exclude this case). 
Let $q_5$  be the cardinality of the set $Q_5:=S_{ca} \setminus (S_{ab} \cup S_{bc})$. We have $c_j \ge b_j \ge a_j$ with at least one strict inequality for every index $j \in Q_5$ (the case $a_j=c_j=b_j$ is excluded). 
Let $q_6$ be the cardinality of the set $Q_6:=S_{bc} \setminus (S_{ab} \cup S_{ca})$. We have $b_j \ge a_j \ge c_j$ with at least one strict inequality for every index $j \in Q_6$ (the case $a_j=c_j=b_j$ is excluded).

As previously mentioned, the number of indices in $S_{bc} \cap S_{ab} \cap S_{ca}$ is zero. The set of remaining indices (i.e., out of $S_{bc} \cup S_{ab} \cup S_{ca}$) is denoted by $Q_7$ and its cardinality is denoted by $q_7$.           


Since we assumed that $a$ does not $\tau$-cover $c$, we know that the number of projects where $a_j \ge c_j$ holds is at most $k$. In other words, $q_3+q_4+q_6+q_7 \le k$. Similarly, $q_1+q_4+q_5+q_7 \le k$ and $q_2+q_5+q_6+q_7 \le k$ hold because $c$ does not $\tau$-cover $b$, and $b$ does not $\tau$-cover $a$, respectively. 
\begin{eqnarray} \label{ineq01}
q_3+q_4+q_6+q_7 &\le& k\\ \label{ineq02}
q_1+q_4+q_5+q_7 &\le& k\\ \label{ineq03}
q_2+q_5+q_6+q_7 &\le& k
\end{eqnarray}

Since we assumed that $a$ $\tau$-covers $b$, we know that there are at least $k+1$ indices such that $a_j \ge b_j$. 
In other words, $q_1+q_3+q_4+q_7 \ge k+1$. Similarly,
$c$ $\tau$-covers $a$ gives $q_1+q_2+q_5+q_7 \ge k+1$, and 
$b$ $\tau$-covers $c$ gives $q_2+q_3+q_6+q_7 \ge k+1$.
\begin{eqnarray} \label{ineq04}
q_1+q_3+q_4+q_7-1 &\ge& k\\ \label{ineq05}
q_1+q_2+q_5+q_7-1&\ge& k\\ \label{ineq06}
q_2+q_3+q_6+q_7-1 &\ge& k
\end{eqnarray}
Combining Inequalities $(\ref{ineq04})$ and $(\ref{ineq01})$ gives 
$q_1+q_3+q_4+q_7-1 \ge k \ge q_3+q_4+q_6+q_7 \Leftrightarrow 
q_1-1\ge q_6$. Since $q_6 \ge 0$ ($q_6$ is the cardinality of a set), we get that $q_1 \ge 1$. Therefore, the instance must contain at least one project index, say 1 w.l.o.g., such that 
\begin{equation} \label{project1}
c_1^g > a_1^r > b_1^y.
\end{equation}
The letters $r$, $g$ and $y$ in the exponent are colors (namely, red, green, and yellow) which are ignored at this stage of the proof; they will be explained and used later on. 
Inequalities $(\ref{ineq05})$ and $(\ref{ineq02})$ give $q_2 -1 \ge q_4$ which indicates the presence of at least one project index, say 2, such that 
\begin{equation} \label{project2}
b_2^r > c_2^y > a_2^g.
\end{equation}
Similarly, Inequalities $(\ref{ineq06})$ and $(\ref{ineq03})$ give $q_3 -1 \ge q_5$ which indicates the presence of at least one project index, say 3, such that 
\begin{equation} \label{project3}
a_3^y > b_3^g > c_3^r.
\end{equation}
Apart from indices $\{1,2,3\}$ described in (\ref{project1}),(\ref{project2}) and (\ref{project3}),  exactly $m-3=2k+1-3=2(k-1)$  project indices remain. It is an even number so we can create pairs between the remaining indices. Since $q_1-1 \ge q_6$, each project index in $Q_6$ can be paired with another project index of $Q_1$ (and different from index 1). This gives us $q_6$ pairs $(j_1,j_1')$ such that 
\begin{equation} \label{project16}
c^r_{j_1} > a^y_{j_1} > b^g_{j_1} \text{ and }
b^r_{{j_1}'} \ge a^g_{{j_1}'} \ge c^y_{{j_1}'}.
\end{equation}
Similarly, $q_2-1 \ge q_4$ says that each project index in $Q_4$ can be paired with another project index of $Q_2$ (and different from index 2). This gives us 
$q_4$ 
pairs $({j_2},{j_2}')$ such that 
\begin{equation} 
b^r_{j_2} > c^y_{j_2} > a^g_{j_2} \text{ and }
a^r_{{j_2}'} \ge c^g_{{j_2}'} \ge b^y_{{j_2}'}.
\end{equation}
Finally, $q_3-1 \ge q_5$ says that each project index in $Q_5$ can be paired with another project index of $Q_3$ (and different from index 3). This gives us 
$q_5$   
pairs $({j_3} ,{j_3} ')$ such that 
\begin{equation} 
a^r_{j_3} > b^y_{j_3}  > c^g_{j_3} \text{ and }
c^r_{{j_3}'} \ge b^g_{{j_3}'} \ge a^y_{{j_3}'}.
\end{equation}
We can continue the pairing of indices 
but at this stage every element of $Q_4\cup Q_5 \cup Q_6$ is matched with an element of $Q_1\cup Q_2 \cup Q_3$.  Only some elements of $Q_1\cup Q_2 \cup Q_3 \cup Q_7$ possibly remain, and the total number of remaining 
indices is even. 


Since $q_1-1 \ge q_6$, there exists a non-negative integer $\delta_1$ satisfying $q_1=\delta_1+q_6+1$ where $\delta_1$ is the number of remaining indices of $Q_1$. Indeed, we have considered one element of $Q_1$ at $(\ref{project1})$, and $q_6$ others at $(\ref{project16})$. Similarly, use $q_2-1 \ge q_4$ and $q_3-1 \ge q_5$ to find two non-negative integers $\delta_2$ and $\delta_3$ satisfying $q_2=\delta_2+q_4+1$ and $q_3=\delta_3+q_5+1$, respectively. The meaning of $\delta_2$ and $\delta_3$ is the same as for $\delta_1$, but they are related to $Q_2$ and $Q_3$, respectively.

Inequality (\ref{ineq03}) can be rewritten as 
$q_2+q_5+q_6+q_7 \le ((2k+1)-1)/2$. Since $m=2k+1=q_1+q_2+q_3+q_4+q_5+q_6+q_7$, the previous inequality becomes 
$q_2+q_5+q_6+q_7 \le (q_1+q_2+q_3+q_4+q_5+q_6+q_7-1)/2 
\Leftrightarrow  \delta_2+q_4+1+q_5+q_6+q_7  \le  (\delta_1+2q_6+\delta_2+2q_4+\delta_3+2q_5+q_7+2)/2  
\Leftrightarrow  \delta_2+q_7  \le  \delta_1+\delta_3$ which gives
\begin{equation}
\label{monlabel1}
\delta_2  \le  \delta_1+\delta_3
\end{equation}
where $q_1=\delta_1+q_6+1$, $q_2=\delta_2+q_4+1$ and $q_3=\delta_3+q_5+1$  have been used. With similar arguments, we can use Inequalities (\ref{ineq01}) and  (\ref{ineq02}) to get that $\delta_1 \le \delta_2+\delta_3$, and $\delta_3 \le \delta_1+\delta_2$.

Let us suppose that $\delta_2 \ge \delta_3 \ge \delta_1$ (the fact that this choice is made w.l.o.g. is explained afterwards). The pairing of the remaining indices is done as depicted on Figure \ref{fig2}. 

\begin{figure}
\begin{center}
\begin{tikzpicture}[scale=0.78]

\draw (0,0) node{$\triangle$} ;
\draw (1,0) node{$\triangle$} ;
\draw (2,0) node{$\triangle$} ;
\draw (3,0) node{$\square$} ;
\draw (4,0) node{$\triangle$} ;
\draw (5,0) node{$\square$} ;
\draw (6,0) node{$\triangle$} ;
\draw (7,0) node{$\star$} ;
\draw (8,0) node{$\star$} ;
\draw (9,0) node{$\star$} ;
\draw (10,0) node{$\star$} ;
\draw (0,1) node{$\circ$} ;
\draw (1,1) node{$\circ$} ;
\draw (2,1) node{$\circ$} ;
\draw (3,1) node{$\circ$} ;
\draw (4,1) node{$\circ$} ;

\draw[dashed]  (-0.3,-0.3) rectangle (0.3,1.3);
\draw[dashed]  (0.7,-0.3) rectangle (1.3,1.3);
\draw[dashed]  (1.7,-0.3) rectangle (2.3,1.3);
\draw[dashed]  (2.7,-0.3) rectangle (3.3,1.3);
\draw[dashed]  (3.7,-0.3) rectangle (4.3,1.3);

\draw[dashed]  (4.7,-0.3) rectangle (6.3,0.3);
\draw[dashed]  (6.7,-0.3) rectangle (8.3,0.3);
\draw[dashed]  (8.7,-0.3) rectangle (10.3,0.3);

\end{tikzpicture}

\end{center}

\caption{\label{fig2} Remaining indices of $Q_2$, $Q_3$,  $Q_1$ and $Q_7$ are represented as circles, triangles, squares, and stars, respectively. The dashed rectangles give the proposed matching}
\end{figure}
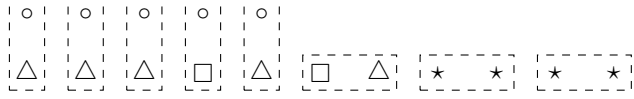

The first horizontal line of Figure \ref{fig2} consists of all the remaining indices of $Q_2$ (each such index is represented by a circle). The second horizontal line of Figure \ref{fig2} consists of 
$\delta_3-\delta_1$ 
remaining indices of $Q_3$ (represented by triangles), followed by 
$\delta_1$ alternations of remaining indices of $Q_1$ (represented by squares) and $Q_3$, followed by all the indices of $Q_7$ (represented by stars). Each index of the first line is matched with the index situated just below in the second line. Since $\delta_2 \le \delta_1+\delta_3$ (cf. (\ref{monlabel1})), we know that every element of the first line has a $Q_1$ or a $Q_3$ element just below. The matching is completed by pairing consecutive indices of the second line. Note that we never match together two elements of the same set, except for $Q_7$. In particular, two $Q_3$ elements of the second line cannot be matched together because $\delta_2 \ge \delta_3$.  Globally, a perfect matching is obtained because the number of remaining indices of $Q_1\cup Q_2 \cup Q_3 \cup Q_7$ is even. 

The choice of $\delta_2 \ge \delta_3 \ge \delta_1$ is w.l.o.g. because if we had $\delta_2 \ge \delta_1 \ge \delta_3$, then we could interchange the roles of $Q_3$ and $Q_1$ in the construction depicted on  Figure \ref{fig2}.
Moreover, $\delta_1 \le \delta_2+\delta_3$ or $\delta_3 \le \delta_1+\delta_2$ could be used instead of $(\ref{monlabel1})$ if $\delta_1$ or $\delta_3$ was the largest element of $\{\delta_1,\delta_2,\delta_3\}$.

In the matching depicted on Figure \ref{fig2}, each index of $Q_2$   matched with an index of $Q_3$ gives a pair $(i_1,i'_1)$ such that
\begin{equation}b^r_{i_1} >c^y_{i_1}> a^g_{i_1}   \text{ and }
a^y_{i'_1} >b^g_{i'_1}> c^r_{i'_1}.
\end{equation}
Each index of $Q_2$ matched with an index of $Q_1$ gives a pair $(i_2,i'_2)$ such that
\begin{equation}b^r_{i_2} >c^g_{i_2}> a^y_{i_2}  \text{ and }
c^y_{i'_2} >a^r_{i'_2}> b^g_{i'_2}.
\end{equation}
Each  index of $Q_1$ matched with an index of $Q_3$ gives a pair $(i_3,i'_3)$ such that
\begin{equation}
c^r_{i_3} > a^g_{i_3}> b^y_{i_3} \text{ and }
a^y_{i'_3} >b^r_{i'_3}> c^g_{i'_3}.
\end{equation}
Each index of $Q_1$ matched with an index of $Q_7$ gives a pair $(i_4,i'_4)$ such that
\begin{equation}
c^r_{i_4} > a^g_{i_4}> b^y_{i_4} \text{ and }
a^y_{i'_4} =b^r_{i'_4}= c^g_{i'_4}.
\end{equation}
Each index of $Q_3$ matched with an index of $Q_7$ gives a pair $(i_5,i'_5)$ such that
\begin{equation}
a^r_{i_5} > b^g_{i_5}> c^y_{i_5} \text{ and }
a^y_{i'_5} =b^r_{i'_5}= c^g_{i'_5}.
\end{equation}
Finally, each index of $Q_7$ matched with another index of $Q_7$ gives a pair $(i_6,i'_6)$ such that
\begin{equation} \label{project77}
a^r_{i_6} = b^g_{i_6}= c^y_{i_6} \text{ and }
a^y_{i'_6} =b^r_{i'_6}= c^g_{i'_6}.
\end{equation}

So far all the indices of $[m]$ have been grouped: one triple (cf. indices $1$, $2$ and $3$ at (\ref{project1}), (\ref{project2}), and (\ref{project3})) and $k-1$ pairs. Each time a color has been given to every agent's declaration. For every project $j$, a bijection between  $\{a_j,b_j,c_j\}$ and $\{r,g,y\}$ has been provided (cf. the letters in the exponents) where $r$, $g$ and $y$ correspond to red, green and yellow, respectively. The matrix is reshuffled so that every line is monochromatic, and corresponds to a potential solution. 

We can observe that every line of the reshuffled matrix is a solution which locally satisfies at least $k+1$ projects of every agent. Let us give the details of this fact for the green one. Following (\ref{project1}), (\ref{project2}), and (\ref{project3}),  the first 3 values of the green solution are $c_1$, $a_2$, and $b_3$. Thus, the green solution locally satisfies agents $a$, $b$ and $c$ for project 1. It locally satisfies agent $a$ for project 2, and agents $b$ and $c$ for project 3. Therefore, every agent in $\{a,b,c\}$ is locally satisfied by the green solution for exactly 2 projects out of the first three. Concerning the $k-1$ pairs whose form is described in cases $(\ref{project16})$ to $(\ref{project77})$, one can verify that the green solution locally satisfies at least one project of each pair, for every agent in $\{a,b,c\}$. Therefore, the green solution locally satisfies $2+k-1=k+1$ projects for every agent, meaning that it $\tau$-covers every agent. 

The same verification can be done for the red solution and the yellow solution. Therefore, the reshuffled matrix is composed of three lines, each one being a solution which satisfies agents $a$, $b$, and $c$.   At least one of the three solutions (red, green and yellow) must be feasible (i.e., the total sum of its coordinates does not exceed 1) because of (\ref{constantsum3}).

A polynomial algorithm can be derived from the previous constructive proof of existence by
simply following the steps of the proof. There are seven subgroups of indices, and their pairings, together with the colors corresponding to the three different solutions, are given in the twelve (in)equalities (\ref{project1}) to (\ref{project77}). Eventually,  output one of the three solutions whose sum of coordinates does not exceed the budget of 1.\end{proof}

\section{Maximizing the Number of Times that Agents are Locally Satisfied} \label{sec:dp:util}

In this section, we consider the problem of maximizing the number of times that agents are locally satisfied. This utilitarian approach outputs a vector that maximizes $f_{util}: {\bf x} \mapsto \sum_{i \in N} \sum_{j=1}^m {\bf 1}(x_j \ge \dell^i _j)$. Here, ${\bf 1}(A)$ is an  
operator which is equal to $1$ when the assertion $A$ is true, 

\begin{theorem} A solution maximizing  $f_{util}$ can be built in polynomial time.  
\end{theorem}
\begin{proof}
For every positive integer $p$, let  $[p]_0=\{0,1, \ldots,p\}$.  
Given $j \in [m]$ and an integer $k \in [jn]_0$, let $\cost(j,k) \in [0,1] \cup \{2\}$ be the minimum ``investment'' for locally satisfying at least $k$ projects of the $n$ agents, but only considering projects from $1$ to $j$. Here, the value $2$ simply means that the 
quantity  
exceeds 1, i.e., it is infeasible, and 2 could be replaced by any different value strictly larger than 1. When $j$ is fixed, $\cost(j,k)$ is a non-decreasing function of $k$. 

For $(j,k) \in [m] \times [n]_0$, let $z_j(k)$ be the minimum value of the $j$-th coordinate of a solution ${\bf x}$ for locally satisfying at least $k$ agents of $N$ for project $j$. The quantity $z_j(k)$ is the 
smallest element $e$ of $\{0\} \cup \{\dell_j^i : i \in N\}$ such that $|\{i \in N : \dell^i_j \le e\}|\ge k$. It holds that  $0=z_j(0)\le z_j(1) \le \cdots \le z_j(n) \le 1$ for all $j \in [m]$. 

The basis of the dynamic programming is $\cost(1,k)=z_1(k)$ for $k=0$ to $n$. Then, use the following recursive formula for $j>1$ and $k=0$ to $jn$: 
\begin{equation*}
\cost(j,k)=\min_{\substack{(a,b) \text{ such that } a+b=k\\ \text{ and } a \ge 0 \text{ and } b\in [n]_0}} \left( \cost(j-1,a)+z_j(b) \right)
\end{equation*}

The justification of this formula is direct: when we are restricted to the $j$ first projects, 
the minimum cost for 
 locally satisfying at least $k$ agent-project pairs is done by locally satisfying at least $b$ agents for project $j$, plus at least $a$ agent-project pairs for projects from $1$ to $j-1$, where $0 \le b \le n$ and $0 \le a \le (j-1)n$.     

Computing an entry $\cost(j,k)$ requires at most $n+1$ elementary operations (one per choice of $b$). In total, $\cost$ has $(n+1)\frac{m(m+1)}{2}$ possible entries. Thus,  $(n+1)^2\frac{m(m+1)}{2}$ operations are necessary, which is polynomial in $n$ and $m$.

Once all the entries of $\cost$ are computed, the largest $k^*$ such that $\cost(m,k^*) \le 1$ corresponds to the maximum value taken by $f_{util}$. The associated solution, say ${\bf x}^*$, can be retrieved by standard backtracking techniques that we skip. 
\end{proof}

\end{document}